\documentclass[journal, letterpaper]{IEEEtran}

\usepackage[utf8]{inputenc}
\usepackage{amsmath,amsthm,amssymb}
\usepackage[compress]{cite}
\usepackage{tabto}
\usepackage{algorithm}
\usepackage[noend]{algpseudocode}
\usepackage{xspace}
\usepackage{paralist}
\usepackage{url}
\usepackage{graphicx}
\usepackage{tabularx,comment,makecell}
\usepackage{dsfont}
\usepackage[dvipsnames]{xcolor}

\usepackage[colorinlistoftodos,textwidth=3cm]{todonotes}

\newtheorem{theorem}{Theorem}

\newtheorem{lemma}{Lemma}

\usepackage{mathtools}
\DeclarePairedDelimiter{\ceil}{\lceil}{\rceil}
\DeclarePairedDelimiter{\floor}{\lfloor}{\rfloor}

\newcommand{\abs}[1]{\left\vert#1\right\vert}
\newcommand{\set}[1]{\left\{#1\right\}}

\newcommand{\eps}{\varepsilon}

\newcommand{\Fig}[1]{Fig.~\ref{fig:#1}}

\newcommand{\Sec}[1]{Sec.~\ref{sec:#1}}
\newcommand{\Tab}[1]{Tab.~\ref{tab:#1}}
\newcommand{\Eq}[1]{(\ref{eq:#1})}
\newcommand{\Alg}[1]{Alg.~\ref{alg:#1}}

\newcommand{\Ic}{\mathcal{I}}
\newcommand{\Ec}{\mathcal{E}}

\newcommand{\Sc}{\mathcal{S}}

\newcommand{\Vc}{\mathcal{V}}

\newcommand{\Gc}{\mathcal{G}}

\newcommand{\rb}{r}
\newcommand{\RSphase}{q}
\newcommand{\RStildephase}{p}
\newcommand{\dload}{T}
\newcommand{\arrival}{a}

\newcommand{\opt}{OPT}
\newcommand{\mopt}{{\text{\opt}}}

\newcommand{\sha}{SHA}
\newcommand{\msha}{{\text{\sha}}}

\newcommand{\highshare}{c-\dhighshare}
\newcommand{\mhighshare}{{\text{\highshare}}}
\newcommand{\dhighshare}{REShare}
\newcommand{\mdhighshare}{{\text{\dhighshare}}}
\newcommand{\packalg}{VM\_assignment}
\newcommand{\mpackalg}{{\text{\packalg}}}
\newcommand{\unpackalg}{VM\_deassignment}
\newcommand{\munpackalg}{{\text{\unpackalg}}}

\begin{document}

\title{Virtual Service Embedding with \\ Time-Varying Load and Provable Guarantees}
\author{
Gil~Einziger,~\IEEEmembership{Member,~IEEE,}
Gabriel~Scalosub,~\IEEEmembership{Senior~Member,~IEEE,}
Carla~Fabiana~Chiasserini,~\IEEEmembership{Fellow,~IEEE,}
Francesco~Malandrino,~\IEEEmembership{Senior~Member,~IEEE}
\thanks{G.~Einziger and G~Scalosub are with the Ben Gurion University of the Negev, Israel. C.~F.~Chiasserini is with Politecnico di Torino, Italy. F.~Malandrino and C.~F.~Chiasserini are with CNR-IEIIT and CNIT, Italy.}
} 
\maketitle
\pagestyle{plain}

\begin{abstract}
Deploying services efficiently while satisfying their quality requirements is a major challenge in network slicing. Effective solutions place instances of the services' virtual network functions (VNFs) at different locations of the cellular infrastructure and manage such instances by scaling them as needed. In this work, we address the above problem and the very relevant aspect of sub-slice reuse among different services. Further, unlike prior art, we account for the services' finite lifetime and time-varying traffic load. We identify two major sources of inefficiency in service management: (i) the overspending of computing resources due to traffic of multiple services with different latency requirements being processed by the same virtual machine (VM), and (ii) the poor packing of traffic processing requests in the same VM, leading to opening more VMs than necessary. To cope with the above issues, we devise an algorithm, called REShare, that can dynamically adapt to the system's operational conditions and find an optimal trade-off between the aforementioned opposite requirements. We prove that REShare has low algorithmic complexity and is asymptotic 2-competitive under a non-decreasing load. Numerical results, leveraging real-world scenarios, show that our solution outperforms alternatives, swiftly adapting to time-varying conditions and reducing service cost by over 25\%.
\end{abstract}

\maketitle

\section{Introduction\label{sec:intro}}

Network slicing leads to a revolutionary transformation of mobile services, with technologies like software-defined networking (SDN) and network function virtualization (NFV) enabling flexible, fully virtualized environments.  
In this context, the automated management of the services that the network supports and of the underlying resources
they consume
is a major challenge. 
As requests for service deployment can arrive and leave at a very fast
pace, it is important that deployment decisions are swift and able to {\em
  dynamically} adapt to the evolution of the network load.
Additionally, each such deployment has to fulfill the services' target key performance indicators (KPIs) and efficiently allocate the very diverse, geographically distributed, and differently owned resources. 

Such requirements imply that when a service request from a third-party
{\em vertical} (e.g., automotive industry or a content provider) 
reaches the system, the following steps should be automatically performed:
\begin{inparaenum}[(i)]
\item to identify the  network segment (e.g.,  edge, aggregation, cloud) where to deploy the service, based on the service target latency  and the infrastructure cost,
\item to check whether part (or all) of the virtual network functions (VNFs) composing the service can re-use already deployed  VNFs (i.e., (sub-)slices),
\item if so, to scale the amount of resources allocated for the VNFs composing such (sub-) slices, so as to fulfill the target KPIs of all involved services, and
\item to instantiate the VNFs that have to be deployed {\em ex novo}, and
  allocate a suitable amount of resources (e.g., CPU, memory, routers) for their processing and interconnection.
\end{inparaenum}

Such a process, also referred to as {\em service orchestration} or {\em  service embedding}, entails multiple, inter-dependent decisions
including VNF placement and virtual machine (VM) provisioning.
In making such decisions, the orchestrator needs to avoid two main sources of inefficiency.
The first  is due to placing VNFs with different delay constraints on the same VM:
in such a case, the most stringent delay constraint is effectively maintained also for the least demanding VNF, which implies that some of the processing allocated to handle the other VNF(s) is not minimal.
Such additional capacity can be seen as derived from latency {\em dissimilarity} of {\em co-located} VNFs and it 
is removed when all those placed in a VM have exactly the same delay constraint. The second inefficiency is due to assigning (``packing'') VNFs in the VMs in a sub-optimal manner,
that is, using more than the minimal number of VMs required.

Several works have addressed these aspects,
including~\cite{noi-satyam,noi-jorge,tulino,tulino2,flexshare,liu2017dynamic,moualla2019online}, by
casting them into MILPs and proposing effective
heuristics. Most of the existing studies, however, operate {\em offline}, i.e., assume that the set of service requests to deploy
is known in advance and make the best (e.g., cost-effective) decisions
for their instantiation.
Importantly, our work is able to capture the effect of time-varying traffic conditions, i.e., with service instances exhibiting finite lifetime and time-varying traffic load and the flexible relationship between the computational resources assigned to a VNF and the resulting service times.

Further, unlike prior art~\cite{tulino,tulino2,liu2017dynamic,moualla2019online}, our work accounts for the fact that next-generation networks are made of segments (e.g., cloud and edge) with different computational capability, latency, and cost: in this scenario, it is important that each service is supported by the segment with the most appropriate cost-performance trade-off.
Furthermore, unlike state-of-the-art works like~\cite{he2018s,feng2017approximation}, our system model and algorithms can account for the  dynamic nature of the load networks have to serve.
Such aspects are
especially critical, as they influence a solution's ability 
to cope with daily and weekly traffic patterns and the network topology.

In this paper, we make the following main contributions:
\begin{itemize}

\item we first develop a model that captures the main aspects of
  the NFV ecosystem 
and accounts for a {\em time varying system load}, as well as for the
fact that there are different  segments composing the network;

\item we formulate the problem of service embedding taking the
  end-to-end latency as the main service KPI, and develop an
  algorithm, named \dhighshare, which  addresses all issues listed above 
 -- network segment identification, VNF placement/reuse,
  and resource scaling --  and, importantly, handles a time-varying service requests load;

\item we prove that \dhighshare has low, namely, quadratic, computational complexity and is an {\em asymptotic 2-competitive algorithm} under a non-decreasing load;

\item finally, using real-world load traces, we show that  the
  cost of deploying and running services under \dhighshare\ is
  much lower than under state-of-the-art solutions. 
  Also, \dhighshare\ can swiftly adapt to time-varying conditions, attaining excellent performance as the system load evolves.
\end{itemize} 

The rest of the paper is organized as follows.
\Sec{model} introduces the system model, while \Sec{RM} outlines our approach and main results.
\Sec{HS} presents a heuristic strategy, which is our cornerstone for building the \dhighshare\ algorithm, and analyzes its competitive ratio.
\dhighshare, along with its competitive analysis, is introduced in \Sec{DHS}, while its performance is assessed in \Sec{results}.
After discussing the related work in \Sec{rel-work}, we conclude the paper in  \Sec{conclusions}.

\begin{table}[t!]
\centering
\caption{Table of notation}
\label{tab:notation}
\begin{tabular}{|c|l|}
\hline 
Symbol &  \multicolumn{1}{c|}{Description} \\
\hline \hline
$\Sc$,\, $\Vc$& set of services and  VNFs (resp.)\\
\hline
$\Vc^s_{\rb}$ & set of VNFs of request $\rb$ for service $s$\\
\hline
$\arrival_{\rb}$, $\tau_{\rb}$  &  arrival and duration (resp.) time of service instance request $\rb$\\
\hline
$\lambda_{\rb}$ &  load of service instance request $\rb$ \\
\hline
$D^s_{\rb}$ & request $\rb$ delay constraint \\
\hline
$\Gc$ & graph of layered topology of datacenters (nodes) \\
\hline
$\ell$ & layer of a node (distance from leaf) \\
\hline
$b$ & VM running in a node \\
\hline
$\bar{\mu}$, $\mu$ & maximum and actual (resp.) computing capability of a VM \\
\hline
$\theta_v$ & computing complexity of VNF $v$ \\
\hline
$l_{\rb}$ & leaf node where request $\rb$ arrives \\
\hline
$\Lambda(b)$ & overall load on VM $b$ \\
\hline
$d_\ell$ & forwarding latency from a leaf to a node in layer $\ell$ \\
\hline
$M_{r,v}$ & minimum processing latency of running $(r,v)$ alone in a VM \\
\hline
$D_r^v$ & fair delay allocation of job $(r,v)$ \\
\hline
$\kappa^{\ell}_f$ & fixed cost of a VM at level $\ell$ \\
\hline
$\kappa^{\ell}_p$ & proportional cost of a VM at level $\ell$ \\
\hline
$\sigma$ & sequence of requests \\
\hline
\end{tabular}
\vspace{-4mm}
\end{table}

\section{System Model}\label{sec:model}

Our system model seeks to concisely represent the main features of next-generation network architectures, namely, their layered structure~\cite{m-layer-1} and heterogeneity. Indeed,
as reported in~\cite{taleb2017multi,vmware,topos} for datacenters and in~\cite{5gg-d23} for 5G networks,
many relevant types of networks are organized in
{\em layers}, with nodes at the same layer having similar features in terms of latency and cost; however, the features of different layers can be vastly different.
In other words, any differences between nodes of the same layer are negligible when compared to those   between nodes of different layers.
By capturing such aspects, our model serves as the first step towards meeting the challenge they represent.
Furthermore, we consider the NFV Orchestrator (NFVO), specified within the
well-known NFV MANO (Management and Orchestration) architecture
\cite{etsi-mano}, as the decision-making entity.
In both ETSI standards~\cite{etsi-mano} and real-world 5G deployments~\cite{5gg-d23}, the NFVO has access to plentiful information concerning such aspects as: (i) the target KPIs of each service, e.g., its target  end-to-end delay; (ii) the time for which each service instance shall be active; (iii) the VNFs composing the service; (iv) their computational requirements, which in turn determine the processing delay.

More specifically, the processing delay associated with a service is determined by (i) the available computational capabilities; (ii) the service computational complexity; and (iii) the service demand. The  computational capabilities available to each service is one of our decision variables, as set forth next. As for the computational complexity of a given service, it is routinely obtained by profiling the service as it runs, hence, it can be considered a known quantity. Finally, the service demand can be either known {\em a priori} (e.g., in smart-factory applications) or reliably predicted~\cite{bega2019deepcog,xu2022esdnn}, as better detailed in \Sec{rel-work}.

We leverage these notions throughout our system model and algorithms, as set forth next. The most relevant notation we use is summarized in Tab.\,\ref{tab:notation}.

{\bf Services.}
Let $\Sc$ be the set of services that verticals may request and $\Vc$  the set of all possible VNFs they may use. Requests for the deployment
of service instances arrive sequentially at the system, and each new request $\rb$ for an instance of service $s$ requires a
set\footnote{We denote the set of VNFs of a service instance request (and the target latency introduced next) by $\Vc^s_{\rb}$ ($D^s_{\rb}$), even if it  depends only on $s$.} 
of connected VNFs $\Vc^s_{\rb} \subseteq \Vc$. 
Such VNFs may either be freshly deployed or reuse existing VNF
instances provided that isolation requirements are met. 

Verticals and network operators commit to service level agreements
(SLAs), detailing the cost and the quantity of the resources each
vertical can use.
Thus, the requests received at the NFVO do not exceed the resource budget for which the vertical
is paying; i.e., there are always enough resources to deploy a service meeting the target KPIs.
Thus, our problem is not whether or not to {\em admit} a service request, but how to {\em embed} it in the most efficient way.

{\bf Network system.}
Most real-world networks connecting datacenters exhibit a layered topology~\cite{taleb2017multi,vmware}, e.g., trees, fat trees, as well as SlimFly and DragonFly~\cite{topos} (an example of such a layered topology appears in \Fig{fat}). 
Let $\Gc=\langle\Ic,\Ec\rangle$ be a graph modelling such a layered topology, with  
 vertices, $i
\in \Ic$, representing datacenters (also referred to as {\em nodes}), and edges, $e \in \Ec$, representing communication links
between datacenters.  
Each node can host a
large~\cite{hong2019resource} number of virtual machines (VMs).
Each VM implements one VNF instance~\cite{5gg-d23}, and the same VNF instance can serve traffic belonging to one or more services.

Layers are numbered as follows: leaf nodes are at layer~$0$, and a generic node
is at {\em layer} $\ell>0$ if its distance in links from the closest leaf is $\ell$.
VMs are characterized by their {\em computational capabilities}, i.e., resources at their disposal. Intuitively, the more capable a VM is, the quicker it will be able to perform a given task. Specifically, a
VM $b$ running in some generic datacenter in node $i$ has maximum speed (i.e.,
computing capability) $\bar{\mu}$ and can host a single VNF $v$.  
We denote with $\mu_b \leq \bar{\mu}$ the amount of computing
resources VM $b$ is using, and
with $\theta_v$ the
{\em computing complexity of VNF $v$}.
Intuitively, computational complexity 
expresses how difficult it is to run a certain VNF. If we fix the amount of computational capabilities, then the processing time required by a given task is proportional to its computational complexity. More formally, computational complexity
is defined as the number of computing  units\footnote{For simplicity,
  $\theta_v\leq 1$ is normalized to the maximum VNF complexity.}
required by  VNF $v$ to handle a  unit of traffic: the higher~$\theta_v$, the more complex the associated VNF, and the more computational resources will be required to run it.

For ease of presentation, we treat the computing capability~$\mu$ as a scalar quantity, i.e., a number; therefore, $\mu$ is able to represent one type of system resources, e.g., CPU. Additional types of resources, like memory or storage, could be accounted for by converting both $\mu$ and the service requirements to multi-dimensional variables, i.e., vectors. Our approach would work unmodified; however, so doing would entail significant additional notation complexity.

A request $r$ is associated with a tuple $\langle\Vc^s_{\rb}, \arrival_{\rb}, \tau_{\rb}, D^s_{\rb}, \lambda_{\rb}, l_{\rb}\rangle$, where $l_{\rb}$ specifies  the leaf node at which the request arrives,
$\arrival_{\rb}$ and $\tau_{\rb}$ are, respectively, its arrival time and duration, $\lambda_{\rb}$ is the traffic load, and $D^s_{\rb}$ is the latency requirement. Having a finite lifetime for requests allows us to model time-varying network load. Specifically, as demonstrated in \Sec{results}, fluctuations in the network load can be reproduced by adding, removing, or replacing service requests.
Notice that the  duration (i.e., lifetime) of a service is a distinct quantity from the time it takes to process an individual packet (or query) belonging to the service, and the quantities often differ by orders of magnitude. As an example, a smart factory service like the one discussed in \Sec{sub-scenarios} may be active for as long as a certain batch has to be produced, i.e., its service duration can be of minutes or days; in the other hand, individual messages pertaining to robot actions are processed in milliseconds.

For each VNF $v \in \Vc^s_{\rb}$, we define $(\rb,v)$ as the  {\em job} of running VNF $v$ for request $\rb$;  each job $(\rb,v)$ has to be assigned to a VM $b$  hosting $v$, running on some node.  
Clearly, the jobs associated with a request  arrive all upon the request's arrival (i.e., at time $\arrival_{\rb}$), and depart after the request's  duration $\tau_{\rb}$ has elapsed (i.e., at time $\arrival_{\rb}+\tau_{\rb}$), when all jobs $(\rb,v)$ are removed and the  resources they consumed are freed. 
Finally,  $\Lambda(b)$ denotes the overall traffic load of the jobs assigned to VM $b$.

\begin{figure}
\centering
\includegraphics[width=.6\columnwidth]{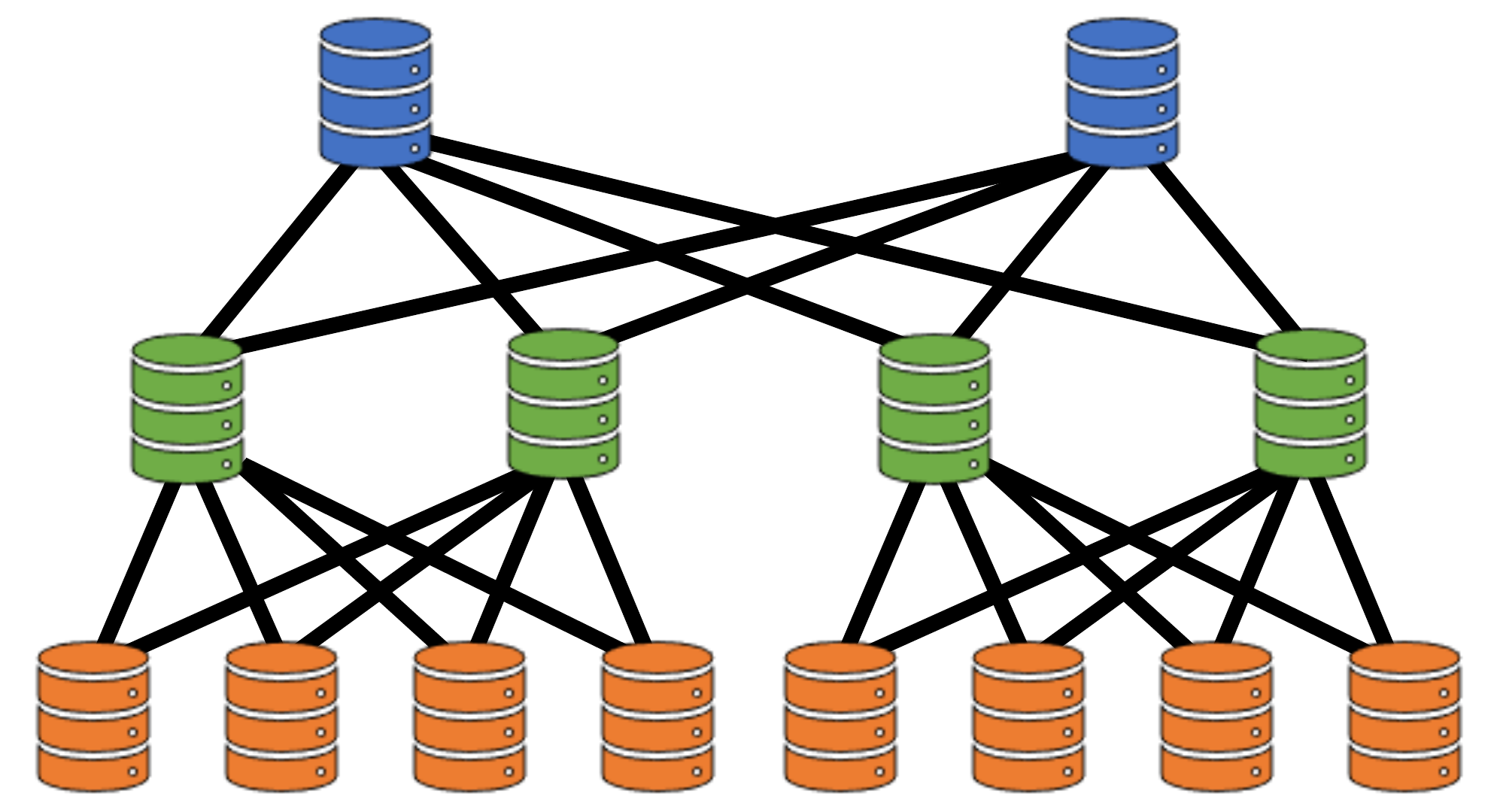}
\caption{
    Example of a layered topology, featuring three layers.
    \label{fig:fat}
} 
\vspace{-4mm}
\end{figure}

{\bf Latency.} 
The total latency incurred by request $\rb$ is given by the sum of the traffic forwarding latency and the processing latency. 
To determine the former, we consider that routes between any two nodes in the topology have been pre-computed. 
We denote with $d_{\ell}$ the traffic forwarding latency from a leaf node to a
node at layer $\ell$; owing to the layered structure of the network
topology (e.g., edge, aggregation, and cloud layers),  we have $d_{\ell+1}>d_{\ell}$,
reflecting the fact that farther-away layers take longer to reach than close-by ones.
We further consider the forwarding latency incurred by request $\rb$, where each job $(\rb,v)$ is deployed at some level $\ell_r^v$, as the {\em maximum} forwarding latency over all levels hosting jobs of $\rb$.
The maximum forwarding latency captures the latency incurred by reaching the highest level, which hosts jobs corresponding to the request.
Such a latency is the dominating component of the overall forwarding latency incurred by the job; indeed, as better detailed in \Sec{HS}, the large number of VMs available in each datacenter makes placing all VNFs of a service in the same node (albeit on different VMs) the most convenient option, hence, inter-VNF forwarding times are not significant.

As for the processing latency, we model each VM  as an $M/M/1$ queue.
Queuing models are commonly used in the literature to represent VMs in edge and cloud scenarios~\cite{neely2008fairness,bi2010dynamic,noi-satyam,prados2021performance,prados2016modeling}, with Markovian service~\cite{bi2010dynamic,noi-satyam,prados2016modeling} and arrival~\cite{neely2008fairness,bi2010dynamic,noi-satyam,prados2016modeling} times being considered in most cases;
the validity of such modeling assumption has also been confirmed through simulation-based validation~\cite{prados2021performance}. Using a queueing model allows us to model how VNFs may be assigned different quantities of resources, e.g., CPU, and how such an amount impacts their processing time. This is in contrast with existing works, e.g.,~\cite{tulino,tulino2,moualla2019online,liu2017dynamic}, where VNF requirements are constant over time and no resource allocation decision is possible.

As per~\cite{Kleinrock},
the processing latency incurred on VM $b$ running VNF $v$ at speed $\mu_b$ is
$\frac{1}{\mu_b-\theta_{v}\Lambda(b)}$.
While VMs get assigned virtual cores in practice, it is fair to assume that $\mu_b$'s take on real values: indeed, given that each node corresponds to a datacenter, it contains a very large number of physical cores, and each physical core can accommodate a high number (e.g., up to 32) of virtual cores.
Importantly,  a queuing model allows us to capture the flexible relation between the computational resources assigned to a VNF and the resulting processing times, i.e., VNFs process requests quicker if they have more resources.
For stability, we must have  $\Lambda(b) < \mu_b/\theta_v$. 

A deployment is 
said to be {\em feasible} for request $\rb$ if
the sum of the traffic forwarding latency and processing
latency over all $v \in \Vc^s_{\rb}$, with the latter computed
considering the maximum VM capability $\bar{\mu}$ for all VNFs, does not exceed 
target delay $D^s_{\rb}$. 
Also, for every  $\rb$, let $\ell^*(\rb)$ be the highest layer 
for which running all $v \in \Vc^s_{\rb}$ on  VMs in a node at that layer is a feasible deployment. 
Next, we define the {\em fair per-job delay allocation}  over the
jobs $(\rb,v)$  
as the per-job processing latency budget, which will serve as a guide for placing request's jobs, allocating  computing resources to them, and performing  VMs  sharing across distinct jobs.

The intuition behind the fair delay allocation is to give more resources to more complex VNFs, i.e., those with larger computational requirements~$\theta_v$. Such an allocation is {\em fair}, in that no VNF contributes in a disproportionate manner to the total processing time.
Specifically, since a deployment of request $\rb$ is feasible
if the total processing latency of its jobs is at most $D^s_{\rb} -
d_{\ell^*(\rb)}$, then 
we distribute this latency budget over the distinct jobs  $(\rb,v)$ according to their {\em relative} processing requirements. 
Namely, for each $(\rb,v)$,  the fair per-job delay allocation  is: 
$D^v_{\rb} = \frac{M_{\rb,v}}{\sum_{u \in \Vc^s_{\rb}} M_{\rb,u}} ( D^s_{\rb} -
d_{\ell^*(\rb)} )$, where   $M_{\rb,v}=\frac{1}{\bar{\mu}-\theta_v \lambda_{\rb}}$ is the latency incurred by running $v$ using the maximum VM capability $\bar{\mu}$, 
and 
$\frac{M_{\rb,v}}{\sum_{u \in \Vc^s_{\rb}} M_{\rb,u}}$
accounts for the relative time complexity of $v \in \Vc^s_{\rb}$. It follows that 
a job $(\rb,v)$ corresponding to a VNF $v$ with a larger value of $\theta_v$ will be assigned a larger value of $D^v_{\rb}$.

Consider now request $\rb$ such that job $(\rb,v)$ is deployed on a node at layer $\ell_r^v \leq \ell^*(\rb)$, while satisfying the fair per-job delay allocation.
The overall processing latency of the request is at most: 
\begin{align}
\sum_{v \in \Vc^s_{\rb}} D_r^v
&=
\sum_{v \in \Vc^s_{\rb}} \frac{M_{\rb,v}}{\sum_{u \in \Vc^s_{\rb}} M_{\rb,u}} ( D^s_{\rb} -
d_{\ell^*(\rb)} ) \\
&= D^s_{\rb} -
d_{\ell^*(\rb)}.
\end{align}
By the definition of $\ell^*(r)$ and the fact that forwarding latency is monotone with the node's level, we have that the overall latency of deploying request $\rb$ is at most $D^s_{\rb}$, and thus feasible.
In the following, we will consider feasible 
service deployments that satisfy such fair per-job delay allocation, and we will compare our solution to an optimal solution that must also satisfy such fair per-job delay bounds.

{\bf Cost.}
Running a VM
$b$ at a node of layer $\ell(b)$ implies a {\em fixed} cost, $\kappa^{\ell(b)}_f$.
The fixed cost reflects the fact that unused VMs still entail additional work for the hypervisor and keep resources allocated, both of which consume power~\cite{kumar2018renewable,morabito2015power}, hence, money.
Furthermore, any VM $b$ running on a node at layer
$\ell(b)$ at speed $\mu_b$ incurs cost $\kappa^{\ell(b)}_p \mu_b$,
which is proportional to the computing resources consumed by  $b$.
Our model accounts for the fact that running VMs at a higher layer
(closer to the cloud) is cheaper \cite{5gt-d14} than running VMs at
lower layers (closer to the edge), i.e., $\kappa^{\ell+1}_f <
\kappa^{\ell}_f$ and $\kappa^{\ell+1}_p < \kappa^{\ell}_p$.
Cost parameters can also account for differently-owned and/or federated resources,
e.g., in different administrative domains.

Denoting by $\sigma$ the sequence of the requests handled by the system till
the current time,  we define $\phi(A,\sigma)$ as the overall {\em
  instantaneous} cost of the deployment decisions made by an algorithm 
$A$, i.e., the sum of fixed and proportional costs of all VMs currently used.
$\phi(A)$, instead, refers to the instantaneous cost of
algorithm $A$ at the end of the whole request arrival/departure process.

\subsection{Problem formulation}
\label{sec:formulation}
The main decisions we have to make are (i) whether VM~$b$ should be used for job~$(\rb,v)$, expressed through binary variables~$y(\rb,b,v)\in\{0,1\}$, and (ii) how much computational resources to assign to each VM, expressed through the real variable~$\mu_b$.
Importantly, $y$-variables also express whether or not a certain VM~$b$ is active; specifically, $b$~is active -- hence, the corresponding fixed cost is incurred -- if and only if it is used by at least one job, i.e.,~$\sum_{(\rb,v)}y(\rb,b,v)\geq 1$.
Also notice that, in our system model, VMs do not migrate, hence, the values of the $y$-variables do not change during the lifetime of a job. 

Our goal is to minimize the total cost, i.e.,
\begin{equation}
\label{eq:obj}
\min_{y,\mu_b} \sum_{b} \left(\kappa_f^{\ell(b)}\max_{(\rb,v)}y(\rb,b,v)+\kappa_p^{\ell(b)}\mu_b \right).
\end{equation}
The two terms in \Eq{obj} correspond to the fixed and proportional cost, respectively.
%
The constraints we have to satisfy are that all target delays are met: for each job~$(\rb,v)$, if~$b$ is the VM serving it, i.e., $y(\rb,b,v)=1$, then we wish to impose
$\frac{1}{\mu_{b} - \theta_v \Lambda(b)} \leq D^v_{r}$.
This is captured by requiring
\begin{align}
D^v_{r} \cdot y(\rb,b,v) \left( \mu_{b} - \theta_v \Lambda(b) \right) & \geq y(\rb,b,v) & \forall (\rb,v),
\end{align}
where $\Lambda(b)=\sum_{\rb'} y(\rb',b,v) \lambda_{r'}$.
We further impose that
\begin{align}
\sum_{b}  y(\rb,b,v) & \geq 1, & \forall (\rb,v), \label{eq:every_job_deployed}\\
\sum_v \max_r y(\rb,b,v) &\leq 1, & \forall b, \label{eq:single_vnf_per_vm} \\
\mu_b &\leq \bar{\mu}, & \forall b, \label{eq:per_vm_capacity_constraint}
\end{align}
where Eq.~(\ref{eq:every_job_deployed}) ensures that every job is actually deployed, Eq.~(\ref{eq:single_vnf_per_vm}) mandates 
that each VM runs a single VNF, and Eq.~(\ref{eq:per_vm_capacity_constraint}) ensures that the capacity constraint of each VM is met.
Since the problem is combinatorial in nature and non-linear, and the formulation is non-convex, we 
adopt a heuristic approach in designing an efficient and effective algorithm for solving the problem.

\section{Approach and Main Results}
\label{sec:RM}
We now describe the approach we use in designing our algorithm, \dhighshare,
which performs feasible VNF placement with {\em time-varying service demand}. \dhighshare\ aims at minimizing the overall cost, 
and it is shown to be  
{\em asymptotically 2-competitive} when requests duration, $\tau_r$, is unlimited.

\dhighshare\ tackles the sources of inefficiency described in \Sec{intro}, as follows:
\begin{itemize}
    \item to reduce the number of requests served before their deadline (and the resulting waste or resources), \dhighshare\ seeks to minimize the {\em dissimilarity} between jobs that are co-located in the same VM;
    \item to reduce the bin-packing suboptimality, \dhighshare\ leverages efficient, state-of-the-art bin-packing algorithms.
\end{itemize}

Our algorithm carefully manages these challenges according to the overall system load.
Intuitively, we pack together diverse services at times of low system load to avoid creating many underutilized VMs. For high system loads, instead, we want to utilize individual VMs better and allow for minimal dissimilarity. 

Our solution is based on a carefully designed algorithm, constant-\dhighshare (\highshare), with parameter $\eps$, described in \Sec{HS}.
The value of $\eps$ governs the amount of dissimilarity allowed; smaller $\eps$ values imply less dissimilarity.
\highshare($\eps$) places all jobs corresponding to a request in a single node at the highest level for which the placement is feasible (even if it requires starting a new VM for every job).
\highshare($\eps$) is shown to be {\em asymptotically $(2+\eps)$-competitive} under a non-decreasing load.

\dhighshare\ uses \highshare($\eps$) with dynamically varying values of $\eps$, depending on the overall system load.
When the system load is low, bin-packing suboptimality is the dominant cost factor, and  \dhighshare\ uses a larger value for $\eps$.
When system load increases, bin-packing suboptimality is naturally reduced since any solution must use many VMs for supporting the load.
In this case, \dhighshare\ favors reducing dissimilarity by using a smaller value for $\eps$.
\Fig{PoD} provides an example of these aspects. 
One of the main novel elements of our approach is such an ability to switch between different decision-making strategies (expressed through different values of~$\eps$) as the network load varies. Indeed, recent works~\cite{tulino,tulino2} leverage a forecast of future requests for the next few hours; \cite{liu2017dynamic} accounts for the network load evolution but always follows the same strategy, i.e., the same trade-off between resource consumption and overhead. Finally, \cite{moualla2019online}~does not adapt its decision strategy to the load and deploys multiple replicas of each VNF chain to achieve robustness.

We decide when to transition between  $\eps$  values by maintaining a lower bound on the cost of an optimal solution (which is closely correlated with the overall system load).
This approach also allows us to prove that \dhighshare\ is asymptotically 2-competitive when requests duration, $\tau_r$, is unlimited.

Notice that assuming infinite request duration is only required for our analysis. In practice, \dhighshare\ is inherently designed to deal with time-varying load conditions where requests arrive and leave the system.
We show the effectiveness of our approach through an extensive simulation study, presented in \Sec{results}, which provides further insight into the dynamic behavior of \dhighshare\ and the reasoning underlying our algorithmic approach.

\begin{figure}
\centering
\includegraphics[width=.75\columnwidth]{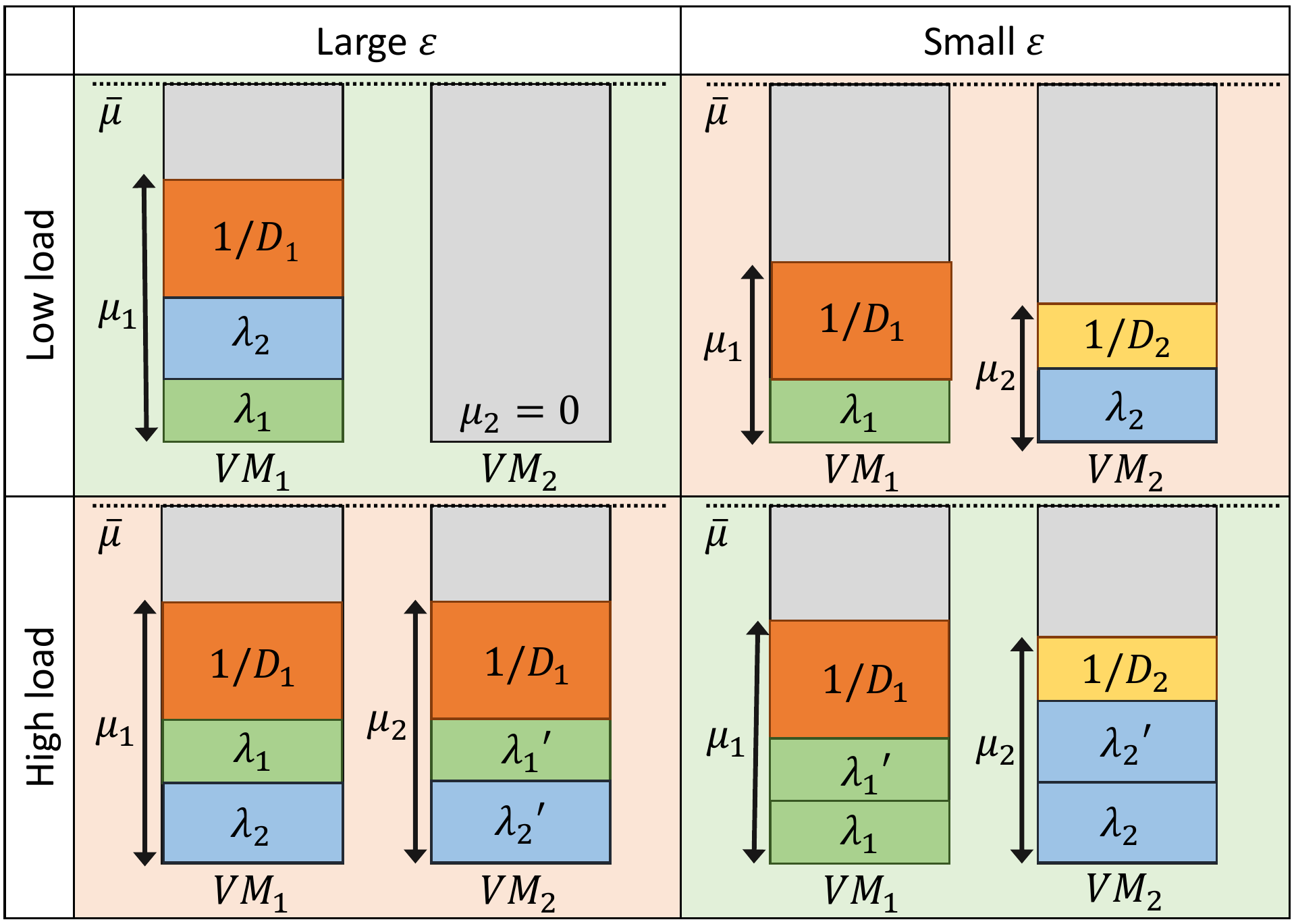}
\caption{Example of the \highshare\ approach: consider  two types of
  services (1 and 2), each composed of the same, single VNF. The
 target latency is $D_1$ and $D_2$ (resp.), with
  $D_1\mathord{<}D_2$. 
As detailed in \Sec{HS}, the capacity required at each VM is given by (i) the total load it must serve (green and blue blocks), plus (ii) the inverse of {\em the shortest delay} it must guarantee (orange and yellow boxes).
Under low load  (one request per service), a large
  $\eps$ (top left) entails that the two requests can share the same
  VM and, considering the M/M/1 modeling, consume
  $\mu_1=\lambda_1+\lambda_2 +1/\min(D_1,D_2)$. Instead, a small $\eps$  (top right) requires two VMs and a larger total capacity. For high load (two requests per service), a large $\eps$ (bottom left) means that services can share the same VMs, but require a high capacity as the most stringent target latency must be met by both VMs. A  small $\eps$ (bottom right) yields a cheaper deployment instead. 
\label{fig:PoD}
} 
\vspace{-4mm}
\end{figure}

\begin{figure*}
\centering
\includegraphics[width=1\textwidth]{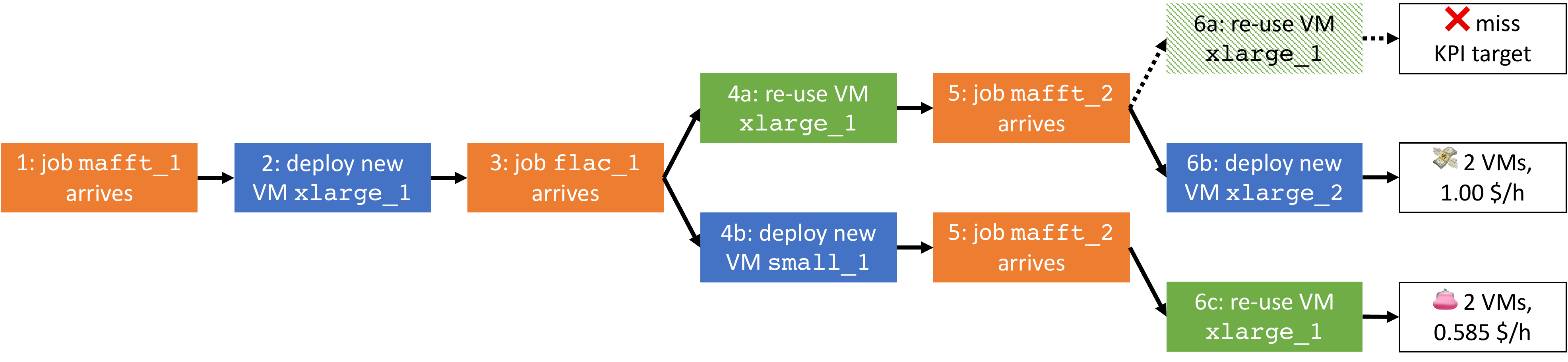}
\caption{
    AWS-based experiment: summary of the possible decisions and their outcome.
    \label{fig:exp-flow}
} 
\end{figure*}

\subsection{An AWS-based experiment}

Before describing our \highshare\ approach, we perform a simple real-world experiment to demonstrate the first, and perhaps less intuitive, of the suboptimality sources that we discuss, namely, serving requests with different deadlines within the same node. We consider two types of requests:
\begin{itemize}
    \item a MAFFT protein alignment task~\cite{mafft} whose target delay is 80~s;
    \item a WAV-to-FLAC audio conversion~\cite{flac} whose target delay  is 50~s.
\end{itemize}
We further consider two types of VMs available on Amazon AWS, namely, \texttt{m1.small} and \texttt{m2.xlarge}. For both VMs, we select the Ubuntu operating system, and run the benchmarks via the \texttt{phoronix} command. When ran alone, the MAFFT benchmark takes 37~s on an \texttt{m2.xlarge} machine, and 137~s on an \texttt{m1.small} one; the FLAC benchmark takes 11~s on an \texttt{m2.xlarge}, and 48~s on an \texttt{m1.small}.

A total of three requests arrive at very short interval from each other, in the following order: first a request for the MAFFT task (\texttt{mafft\_1}), then one for the FLAC task (\texttt{flac\_1}), and finally a second MAFFT request (\texttt{mafft\_2}). Recall that we take an online approach, i.e., we embed all requests as soon as they arrive.

The possible decisions are summarized in \Fig{exp-flow}. Clearly, when request \texttt{mafft\_1} arrives (step~1), we must create a new VM. Since an \texttt{m1.small} VM would be too slow, we must open an \texttt{m2.xlarge} VM, called \texttt{xlarge\_1} (step~2). After job \texttt{flac\_1} arrives (step~3), we can make two decisions: re-use existing VM \texttt{xlarge\_1}, or open a new VM \texttt{small\_1} of type \texttt{m1.small}.

Re-using \texttt{xlarge\_1} (step~4a) would minimize both the number of VMs and the cost {\em so far}, however, things would be different after the arrival of the third job (step~5): using \texttt{xlarge\_1} (step~6a) would result in a violation of the target KPIs, hence, we need to open a new VM.  Furthermore, such a new VM can only be of type \texttt{m2.xlarge} (step~6b), resulting in a higher cost of 1 USD/hour\footnote{All prices refer to the US-east availability zone.}.

Let us assume instead that we create a new VM for the newly-arrived \texttt{flac\_1} job, on the grounds that its target delay and complexity are vastly different from the ones of job \texttt{mafft\_1}; importantly, a cheaper VM of type  \texttt{m1.small} would suffice. Performing this counter-intuitive action allows us to serve request \texttt{flac\_2} on the existing VM \texttt{xlarge\_1} (step~6c), resulting in a total lower cost of 0.585\,USD/hour.

In summary, avoiding serving requests with overly-different values of target delay  within the same VM results in lower global costs. As confirmed in \Tab{exp-data}, it also results in requests served closer to their deadline, hence, intuitively, in less wasted computational capabilities.
\begin{table}
\caption{
    AWS-based experiment: service time and cost for the different options reported in \Fig{exp-flow}
    \label{tab:exp-data}
} 
\centering
\begin{tabular}{|l|l|r|}
\hline
{\bf case} & {\bf service time [s]} & {\bf cost [USD/hour]} \\
\hline\hline
6a &\makecell{ {\color{red}\texttt{mafft\_1}: 85}\\{\color{red}\texttt{mafft\_2}: 85}\\{\color{ForestGreen}\texttt{flac\_1}: 33}} & 0.5\\
\hline
6b & \makecell{ {\color{ForestGreen}\texttt{mafft\_1}: 48}\\{\color{ForestGreen}\texttt{mafft\_2}: 22}\\{\color{ForestGreen}\texttt{flac\_1}: 37}} & 1\\
\hline
6c & \makecell{ {\color{ForestGreen}\texttt{mafft\_1}: 74}\\{\color{ForestGreen}\texttt{mafft\_2}: 74}\\{\color{ForestGreen}\texttt{flac\_1}: 45}} & 0.585\\
\hline
\end{tabular}
\end{table}

\section{The \highshare($\eps$) Strategy}
\label{sec:HS}
This section describes and analyzes the constant-\dhighshare\ (\highshare) strategy, with a fixed parameter $\eps$.

\subsection{Algorithm description}\label{sec:algo}


The details of \highshare($\eps$) are presented in \Alg{highshare},
which takes  $\eps$, determining the level of VNF (hence,
VM) sharing among jobs, as an input
parameter. 
Given a request $\rb$, first 
\highshare($\eps$) identifies layer $\ell^*(\rb)$ (Line\,\ref{alg:hs:highest_level}), which
is the highest  layer where the requested service instance can be deployed
without violating the service KPI targets and keeping the fair per-job delay allocation (see \Sec{model}).
We remark that such a  layer is selected because VNF deployment   is cheaper 
at  higher layers and that  
this layer can be found by conducting a binary search, i.e., in time that is logarithmic in the number of layers.
Furthermore, \highshare($\eps$) deploys all the service VNFs in the same node to minimize the traffic forwarding latency and maximize the computation latency budget.
Recall that, as per our system assumptions, it is {\em always} possible to place in the same node all VNFs composing a service. Furthermore, whenever doing so is possible, it is also beneficial, as it avoids incurring inter-node forwarding latency.

The specific node $i^*$ at layer $\ell^*(\rb)$ where the request
traffic should be processed is chosen as the one providing the best
load balancing (Line\,\ref{alg:hs:highest_node}),
i.e., the one that currently has the lowest load. However, it is important to highlight that the choice of~$i^*$ does not affect the worst-case performance of the algorithm and that we can use alternative criteria (e.g., choosing the most-recently-used or most-loaded server) with similar guarantees.

\begin{algorithm}[t]
\caption[\highshare]{\highshare($\eps$)}
\label{alg:highshare}
\begin{algorithmic}[1]
\For{any arrival/departure of request $\rb$}
\If{$\rb$ is arriving}
    \State $\ell^*(\rb) \gets \max \set{ \ell |\mbox{fair 
        delay allocation is met at}\, \ell }$
        \label{alg:hs:highest_level}
        \Comment{choose highest possible layer}
            \State $i^* \gets \text{Load\_balancing}(G,\ell^*(\rb))$
        \label{alg:hs:highest_node}
        \Comment{select node}
        \For{each $v \in \Vc^s_r$}
            \State $j^* \gets \floor*{ \log_{(1+\eps)} D^v_r}$
            \label{alg:hs:partition}
            \Comment{determine latency range}
            \State $\mpackalg((\rb,v),(i^*,v,j^*)\mbox{-AP})$
             \label{alg:hs:AP}
            \Comment{assign job to a VM}
        \EndFor
    \Else
    \Comment{$\rb$ is departing}
        \For{each $v \in \Vc^s_r$}
        \State $\munpackalg((\rb,v))$
        \Comment{deallocate the job}
        \EndFor
    \EndIf
\EndFor
\end{algorithmic}
\end{algorithm}

Assigning jobs $(\rb,v)$ to VMs can now be done {\em
  independently} over the VNFs, as long as the fair per-job delay
allocation, $D_r^v$, $v \in \Vc_r^s$, is met. 
Placing each job in a new or existing VM running $v$  poses the following dilemma:
\begin{itemize}
    \item on the one hand, we would like to {\em share} VMs as much as possible, to save on fixed costs;
    \item on the other, sharing the same VM among jobs with different
      delay allocation results in wasted capacity 
(as discussed earlier and exemplified in \Fig{PoD}), and, thus, in higher proportional costs.
\end{itemize}
The parameter~$\eps$ describes the tradeoff underlying such decisions.
Specifically, let $\lambda_{\min}=\inf_{\rb}\set{\lambda_r}$, where we assume $\lambda_{\min}>0$ and is known in advance;  
then we define a series of non-overlapping latency ranges $L_j$, where
$L_j=\left( \frac{1}{\bar{\mu} - \lambda_{\min} (1+\eps)^j}, \frac{1}{\bar{\mu} - \lambda_{\min} (1+\eps)^{j+1}} \right ]$,
for $j\geq 1$,
and $L_0=\left[ \frac{1}{\bar{\mu} - \lambda_{\min}}, \frac{1}{\bar{\mu} - \lambda_{\min} (1+\eps)} \right ]$.
Clearly, we need to impose that 
$\lambda_{\min} (1+\eps)^{j+1} < \bar{\mu}$, $\forall j$; then, the number of latency ranges, $J_{\eps}$, satisfies
$J_{\eps} < \log_{(1+\eps)}\frac{\bar{\mu}}{\lambda_{\min}}-1$. 
\highshare($\eps$) only shares the same VM among jobs within the same
latency range. Thus, a larger value of $\eps$ results in wider ranges and jobs with more diverse fair delay allocations sharing the same VM. On the contrary, a
smaller value of $ \eps$ implies narrower ranges and more {\em similar} jobs within each range. 

Given $i^*$, the objective is to assign each VNF $v\in\Vc_r^s$ to a
suitable VM in $i^*$, such that the fair delay allocations are satisfied for all jobs on that VM (Line\,\ref{alg:hs:partition}). 
In particular, given a VNF $v$, we consider the {\em VM
  Assignment Problem}, hereinafter referred to as  $(i,v,j)$-AP, which 
assigns a job  $(\rb,v)$ to a VM in node $i$, while ensuring that $D^v_r \in L_{j}$ (with $j=\floor{ \log_{(1+\eps)} D^v_r}$). Since we only share VMs among jobs in the same range, it is possible to consider separate assignment problems for different values of~$j$.

\begin{algorithm}[t]
\caption[\packalg]{$\mpackalg((\rb,v),(i,v,j)\mbox{-AP})$}
\label{alg:packalg2}
\begin{algorithmic}[1]
\State{$\texttt{viable\_VMs}\gets \emptyset$}
\For{every VM $b$ in $i$ hosting jobs $(r',v)$ in $L_j$}
\label{alg:packalg:for_begin}
 \If{$\frac{1}{\bar{\mu} - \theta_v(\Lambda(b) + \lambda_r)} > D^v_{r}$}
 \label{alg:packalg:check_feasibility_start1}
 \Comment{check feasibility for the new job on $b$}
  \State{{\bf continue}} \Comment{continue to next VM}
 \EndIf
 \For{every $(r',v) \in b$}
 \label{alg:packalg:check_feasibility_start2}
 \Comment{jobs already assigned to $b$}
  \If{$\frac{1}{\bar{\mu} - \theta_v(\Lambda(b) + \lambda_r)} > D^v_{r'}$} 
  \Comment{check feasibility for jobs already on $b$}
   \State{{\bf continue}} \Comment{continue to next VM}
  \EndIf
\label{alg:packalg:for_end}
 \EndFor
\label{alg:packalg:check_feasibility_end2}
\State{$\texttt{viable\_VMs}\gets \texttt{viable\_VMs} \cup \set{b}$} \label{alg:packalg:viable_vms} \Comment{reusable VMs}
\EndFor
\If{$\texttt{viable\_VMs}\neq\emptyset$}
 \State{$b^*\gets\arg\max_{b\in \texttt{viable\_VMs}} \Lambda(b) $} \label{alg:packalg:bestfit} \Comment{Best-fit}
 \State{$\mu_{b^*}\gets \theta_v[\Lambda(b)+\lambda(r)] + \frac{1}{\min_{(\rb',v)\in b}D^v_ {r'}}$}
 \label{alg:packalg:adjust1} \Comment{adjust capability} 
\Else
 \State{$b^*\gets\textbf {create new VM}\text{ in } i$}
 \label{alg:packalg:new}
 \State{$\mu_{b^*}\gets \theta_v \lambda_r + \frac{1}{D^v_r}$}
 \label{alg:packalg:adjust2} \Comment{adjust capability} 
\EndIf
\State{$\textbf{place }(\rb,v)\text{ in }b^*$}
\label{alg:packalg:place}
\end{algorithmic}
\end{algorithm}

The \mpackalg\ procedure for solving $(i,v,j)$-AP is presented in Alg.\,\ref{alg:packalg2}.
It begins by looking for VMs whose capacity could be expanded to
accommodate the additional load of job $(\rb,v)$, while honoring its
delay allocation $D_r^v$ (Line~\ref{alg:packalg:check_feasibility_start1}) as well as that of previously allocated jobs (Lines~\ref{alg:packalg:check_feasibility_start2}-\ref{alg:packalg:check_feasibility_end2}).
It then places such VMs in set \texttt{viable\_VMs}
(Line~\ref{alg:packalg:viable_vms}). If the set is not empty, in
Line~\ref{alg:packalg:bestfit} the viable VM with the least free
capacity is selected, in a Best-fit \cite{epstein09robust} fashion.
If there are no viable VMs, a new VM is instantiated in Line~\ref{alg:packalg:new}.
In either case, job $(\rb,v)$ is assigned to VM~$b^*$ (Line~\ref{alg:packalg:place}).
In Line~\ref{alg:packalg:adjust1} (and similarly in
Line~\ref{alg:packalg:adjust2}), the capacity of~$b^*$ is adjusted to
guarantee that the fair delay allocations  for all jobs assigned to $b^*$ are met.
Through simple manipulation of the M/M/1 delay expression, this
corresponds to setting $\mu_{b^*}=\theta_v[\Lambda(b)+\lambda(r)] +
\frac{1}{\min_{(\rb',v)\in b}D^v_ {r'}}$.  Line~\ref{alg:packalg:viable_vms},
and the fact that the delay allocation  of job $(\rb,v)$ can always be satisfied
by placing it in a new VM,
imply that it is always possible.

The departure of service instances is dealt with in \Alg{unpackalg},
which removes all jobs of request $r$ from the VMs hosting them and
readjusts the VMs capability to release the resources while meeting 
the 
constraints of all remaining jobs.

\begin{algorithm}[t]
\caption[\unpackalg]{$\munpackalg((\rb,v))$}
\label{alg:unpackalg}
\begin{algorithmic}[1]
\State $i \gets$ node running request $\rb$
\State $b \gets$ VM on node $i$ running job $(\rb,v)$
\State remove $(\rb,v)$ from $b$
\State $\mu_b \gets \theta_v \Lambda(b) + \frac{1}{\min_{(\rb',v) \in b}D^v_{r'}}$
\label{alg:unpackalg:adjust_capability}
\Comment{adjust capability}
\end{algorithmic}
\end{algorithm}

\subsection{Competitive ratio analysis}\label{sec:compratio-HS}

In this section, we compare the cost of  \highshare($\eps$) against that of an optimal solution, denoted by \opt, which is aware of the future sequence of service requests that arrive at the NFVO.   
We derive our worst-case performance guarantees for scenarios where each request has infinite duration.
This analysis guides our algorithmic design and provides asymptotic performance guarantees when system load tends to infinity.
However, as we show in later sections, our proposed solutions also provide significantly improved performance in cases where requests have a finite duration, and the system load increases and decreases in a time-varying manner.
We first provide below an overview of our analysis.

{\bf High-level description of the analysis.}
At the heart of our analysis lies a load argument across distinct layers, where we compare the load handled by \highshare($\eps$) at some layer to that handled by \opt.
We consider a shadow solution, \sha($\eps$) (described in the sequel), which is allowed to produce a {\em fractional assignment}, and should also satisfy {\em relaxed delay constraints}.
This solution runs {\em most} of its VMs at full capacity, and the cost of such VMs serves as a {\em lower bound on the cost of \opt}.
We show that the cost induced by the deployment performed by \highshare($\eps$) is comparable to that of this lower bound, and provide an additional bound on the {\em additive} cost of all remaining VMs that may be used by \opt.
In what follows, we provide a detailed account of our analysis.

We first show that \highshare($\eps$) performs feasible deployments for any service request.
\begin{lemma}
\label{lemma:alg:hs:soundness}
If \highshare($\eps$) assigns $\rb$ to node $i^*$ at layer $\ell^*(\rb)$, allocating VM capability $\mu_{b^*}$,  then such a deployment of the service instance is feasible.
\end{lemma}
\begin{proof}
Consider $\rb$ being placed at node  $i^*$, located at layer $\ell^*(\rb)$. 
By the definition of \Alg{highshare}, $\ell^*(\rb)$ is the highest layer where deploying all $v \in \Vc^s_{\rb}$ on new VMs in a node at layer $\ell$ is feasible.
When considering \Alg{packalg2}, each job $(\rb,v)$ is placed in one of the viable VMs at $i^*$ where running the VM at maximum processing capability ensures a feasible solution (Lines~\ref{alg:packalg:for_begin}--\ref{alg:packalg:for_end}).
Since \Alg{packalg2} uses the minimal capability that ensures feasibility (Lines~\ref{alg:packalg:adjust1}--\ref{alg:packalg:adjust2}), the result follows.
\end{proof}

Next,  we compare the performance of \highshare($\eps$) to the optimum. To bound the competitive ratio and leverage load rearrangement arguments in our proofs, we introduce 
 an alternative request (hence, job) arrival/departure process, and an alternative placement strategy. 
Specifically, 
for any job $(\rb,v)$ associated with latency range $L_j$, i.e., such that 
$D^v_r \in \left(  \frac{1}{\bar{\mu} - \lambda_{\min} (1+\eps)^j} , \frac{1}{\bar{\mu} - \lambda_{\min} (1+\eps)^{j+1}} \right]$
(or, possibly,  $D^v_{\rb}=\frac{1}{\bar{\mu}-\lambda_{\min}}$ in the case of $j=0)$, 
we define a corresponding {\em top} job, $\overline{(\rb,v)}$, as
equivalent to $(\rb,v)$ but associated with delay constraint
$\overline{D^v_r} = D_j = \frac{1}{\bar{\mu} - \lambda_{\min}(1+\eps)^{j+1}} \geq D^v_r$.
\footnote{Note that 
the load offered by $\overline{(\rb,v)}$ is the same as that of $(\rb,v)$, and the difference is that the top job has a (possibly) more relaxed  delay constraint.}
Note that  all top jobs falling in the same latency range $L_j$ have the same delay constraint $D_j$.

{\bf A shadow strategy.}
We now introduce the {\em shadow fractional assignment (\sha($\eps$))},  a {\em fractional} placement strategy, whose cost
will be used to determine
a lower bound on the cost of \opt\ within the analysis of \highshare($\eps$).
Furthermore, in \Sec{DHS} we also use \sha($\eps$) explicitly within our algorithm \dhighshare, which can handle time-varying 
load.

\sha($\eps$) uses the same $\eps$ value as \highshare($\eps$), and operates as follows: 
\begin{inparaenum}[(i)]
\item \label{sha_top} it handles the top job sequence $\overline{(\rb,v)}$;
\item it never places in the same VM jobs associated with different $L_j$ (hence, by (i) and the definition of top jobs, with different delay allocations);
\item \label{sha_feasible}
it places a job in a single node; 
\item \label{sha_fractional} within the chosen node, it can place fractions of a job load on different VMs (i.e., it works with a ``fluidified'' version of the jobs);
\item \label{sha_optimal} it generates the {\em optimal} placement that satisfies  conditions~(\ref{sha_top})-(\ref{sha_fractional}) above.
\end{inparaenum}
In case of a request departure, \sha($\eps$) removes the load
associated with the corresponding jobs and fractionally rearranges the
remaining load; thus, 
the cost of the resulting placement is equivalent to that of
re-running \sha($\eps$) on all the remaining requests. 

\sha($\eps$) is used in two manners:
\begin{inparaenum}[(i)]
\item guiding the decisions of our {\em dynamic} algorithm, \dhighshare\ (details are provided in \Sec{DHS}), and
\item identifying a {\em lower bound on the cost of \opt}.
\end{inparaenum}
We now briefly explain how the cost of \sha($\eps$) can be used to define a lower bound on the cost of \opt, where we refer to this property as the {\em relaxation property of \sha($\eps$)}.
Note that by using the top sequence with a
fractional assignment, for each VNF $v$, every latency value $D_j$ of a top job, and every node $i$, all the VMs, save at most one, used by \sha($\eps$) in $i$   work at full computing capacity $\bar{\mu}$ to handle jobs associated with $v$ and $D_j$. 
Since such full VMs handle the load placed upon them in the most efficient manner (i.e., there are no inefficiencies due to placing jobs with different delay constraints on the same VM), this amount of load must also be handled by \opt. All other properties of \sha($\eps$) essentially relax the constraints imposed on \opt. It follows that the cost of running these full VMs serves as a lower bound on the cost of \opt.

We begin our analysis
by showing that \highshare($\eps$) places jobs in the highest layer possible.
\begin{lemma}
\label{lem:alg_places_above_opt}
If \highshare($\eps$) places the jobs of request $\rb$ in node $i^*$ at layer $\ell^*(\rb)$, then \opt\ and \sha($\eps$) both place every job $(\rb,v)$ in some node $i'$ at some layer $\ell'$, with $\ell' \leq \ell^*(\rb)$.
\end{lemma}
\begin{proof}
We start by showing that the claim holds for \opt. Assume, by contradiction, that \opt\ places some job $(\rb,v)$ at node $i'$ at layer $\ell'>\ell$.
Then, the network latency of \opt\ is higher than that of \highshare($\eps$) (since the traffic associated with $\rb$ has to reach layer $\ell'>\ell$), i.e., $d_{\ell'}>d_{\ell}$.
As for the processing latency yielded by the \opt\ placement, this is at least $\sum_{v \in V^s_r} M_{r,v}$.
By construction,   \highshare($\eps$) ensures that node $i^*$ is at the highest layer $\ell^*=\ell^*(\rb)$ for which
$\sum_{v \in V^s_r} M_{r,v} \leq D^s_r-d_{\ell^*}$; 
this implies that the assignment made by \opt\ is not feasible, thus contradicting the initial assumption.
By item~(\ref{sha_feasible}) in the definition of \sha, the above argument also holds for \sha($\eps$).
\end{proof}

In the remainder of this section,
we consider that requests have an infinite duration, i.e., $\tau_r=\infty$, for all arriving requests $r$.
Under this condition, we can bound the total amount of traffic load (over all possible service requests) that \sha($\eps$) may process at a different layer than the one selected by \highshare($\eps$) for any given node used by \highshare($\eps$), any VNF, and any latency range.
\begin{lemma}
\label{lem:max_disagreement_load}
For every node $i$ at layer $\ell$, every VNF $v$, and latency range
$L_j$, the overall load of VNF $v$, handled by \highshare($\eps$) at
$i$ and associated with $L_j$, that is handled by \sha($\eps$) at
a  layer $\ell' \neq \ell$, is at most $n_{\ell'} \cdot \lambda_{\min}(1+\eps)^{j+1}$, where $n_{\ell'}$ is the number of nodes at layer $\ell'$. 
\end{lemma}
\begin{proof}
Assume by contradiction that there exist some VMs running VNF $v$ at layer $\ell'$,  associated with latency range  $L_j$, that, according to \sha($\eps$), handle a load higher than   $n_{\ell'} \cdot \lambda_{\min}(1+\eps)^{j+1}$, while  \highshare($\eps$) handles that load at node $i$ at layer $\ell \neq \ell'$.  
By Lemma~\ref{lem:alg_places_above_opt}, we have that $\ell' < \ell$.  
By the pigeonhole principle~\cite{pigeon}, there exists at least one node at a layer $\ell' < \ell$, such that there exists a total load of at least $\lambda_{\min}(1+\eps)^{j+1}$ corresponding to $L_j$, that is handled by \highshare($\eps$) in node $i$ at layer $\ell$, but is handled by \sha($\eps$) in node $i'$ at layer $\ell'$.
Without loss of generality, we assume that there exists a {\em full} VM in node $i'$ at layer $\ell'$ that, according to \sha, processes  $\lambda_{\min}(1+\eps)^{j+1}$ traffic with latency range $L_j$. This assumption is fair since \sha($\eps$) applies a fractional load assignment and places in the same VM only jobs associated with the same latency range, while the delay constraints of all jobs in that latency range are the same in the top sequence.

Consider now an alternative solution, $\msha(\eps)'$, which is identical to \sha($\eps$), except for having this entire VM run at some node $i''$ at layer $\ell'+1$.
First, note that this produces a feasible shadow fractional assignment for the workload handled by \sha($\eps$) and does not require increasing the processing speed of the VM at layer $\ell'+1$. Indeed, since \highshare($\eps$) places these jobs at layer   $\ell'+1$ or higher, the per-job latency constraint is satisfied for the solution of $\msha(\eps)'$ and, by definition of $L_j$, \sha($\eps$) already processes $\lambda_{\min}(1+\eps)^{j+1}$  traffic at $i'$ using the maximum VM computing capability.
Second, the cost of $\msha(\eps)'$ is strictly smaller than the cost of \sha($\eps$) since, at layer $\ell'+1$, both the proportional cost and the fixed cost are smaller than at layer $\ell'$.
This contradicts the optimality of \sha, thus completing the proof. 
\end{proof}

Next, for every VNF $v$, every node $i$, and every latency range
$L_j$, let $\Lambda^v_{i,j}$  be the 
 total load due to all jobs $(\rb,v)$ 
 that are handled by \highshare($\eps$) in node $i$ at layer $\ell$.
 We next provide an upper bound on the maximum number of VMs that \highshare($\eps$) requires to handle such load,
which is then  used to prove the competitive ratio of \highshare($\eps$).
\begin{lemma}
\label{lem:highshare:number_of_VMs}
Given  VNF $v$,  node $i$, and latency range $L_j$, the number of VMs used by \highshare($\eps$) to handle workload $\Lambda^v_{i,j}$ is at most $\frac{2  \Lambda^v_{i,j}}{\lambda_{\min}(1+\eps)^j} + 1$. 
\end{lemma}
\begin{proof}
First, note that the delay constraint $D^v_r$ of every 
job
$(\rb,v)$
contributing to $\Lambda^v_{i,j}$ is at least $\frac{1}{\bar{\mu}-\lambda_{\min}(1+\eps)^j}$.
It follows that we can pack a load of at least $\lambda_{\min}(1+\eps)^j$, while satisfying the delay constraints of the jobs assigned to the VM.
We therefore view this latter quantity as the {\em lower bound on the size} of the VMs in the $(i,v,j)$-AP.

The overall load on each VM $b$ of $(i,v,j)$-AP, except for maybe one, is at least $\frac{\lambda_{\min}(1+\eps)^j}{2}$.
This argument follows from the fact that if there were two VMs with load strictly less than $\frac{\lambda_{\min}(1+\eps)^j}{2}$, combining their loads on a single VM would have resulted in a VM with an overall load of at most $\lambda_{\min}(1+\eps)^j$. This assignment is still feasible and incurs a lower placement cost.  
Thus, running such a VM at speed $\bar{\mu}$ would result in every job $(\rb,v)$ assigned to the VM experiencing a latency at most 
$\frac{1}{\bar{\mu} - \lambda_{\min}(1+\eps)^j} \leq D^v_r$,
where the inequality follows from the definition of  $L_j$.
This contradicts the definition of \packalg\ (Alg.~\ref{alg:packalg2}), which assigns jobs to already open VMs if the overall load on the VM would still result in a feasible solution (for some speed no larger than $\bar{\mu}$).
Hence, the overall number of VMs used by \highshare($\eps$) for handling  $\Lambda^v_{i,j}$ is at most
$\ceil[\big]{\frac{2 \cdot \Lambda^v_{i,j}}{\lambda_{\min}(1+\eps)^j}} 
\leq \frac{2 \cdot \Lambda^v_{i,j}}{\lambda_{\min}(1+\eps)^j} + 1$,
which proves the thesis.
\end{proof}

\begin{theorem}
\label{thm:REShare_competitive_ratio}
\highshare($\eps$) is a $2(1+\eps)$-asymptotic competitive algorithm when service requests have unlimited duration.
\end{theorem}
\begin{proof}
Consider $\Lambda^v_{i,j}$;  
this can be seen as the sum of loads of two types of jobs: type 1,  corresponding to jobs that are handled by \sha($\eps$) at some node $i'$ at layer $\ell$, with an overall load $\tilde{\Lambda}^v_{i,j}$, and type 2, corresponding to jobs that are handled by \sha($\eps$) at some node at a layer other than $\ell$, with an overall load of $\bar{\Lambda}^v_{i,j}=\Lambda^v_{i,j}-\tilde{\Lambda}^v_{i,j}$.
We denote the amount of resources used by \highshare($\eps$) for handling $\Lambda^v_{i,j}$  
by $\mhighshare^v_{i,j}$, and that used by \sha($\eps$) for handling jobs contributing to $\tilde{\Lambda}^v_{i,j}$ by $\msha^v_{i,j}$.

First, by Lemma~\ref{lem:highshare:number_of_VMs}, we have that the overall number of VMs used by \highshare($\eps$) for handling load $\Lambda^v_{i,j}$ is at most
\begin{align}
\frac{2  \Lambda^v_{i,j}}{\lambda_{\min}(1+\eps)^j} + 1
& = \frac{2 \tilde{\Lambda}^v_{i,j}}{\lambda_{\min}(1+\eps)^j} + \frac{2  \bar{\Lambda}^v_{i,j}}{\lambda_{\min}(1+\eps)^j} + 1 \notag \\
& \leq \frac{2  \tilde{\Lambda}^v_{i,j}}{\lambda_{\min}(1+\eps)^j} +
\frac{2 n  \lambda_{\min}(1+\eps)^{j+1}}{\lambda_{\min}(1+\eps)^j} + 1 \notag \\
& = \frac{2  \tilde{\Lambda}^v_{i,j}}{\lambda_{\min}(1+\eps)^j} +
2 n (1+\eps) + 1\,, \notag
\end{align}
where the inequality follows from Lemma~\ref{lem:max_disagreement_load}.
If we let $\kappa^\ell=\kappa_f^\ell + \kappa_p^\ell \bar{\mu}$ denote the cost of running a single VM at maximum speed in node $i$ at layer $\ell$, then
\begin{equation}
\label{eq:highshare_cost}
 \phi(\text{\highshare}^v_{i,j}) \leq 
\frac{2 \tilde{\Lambda}^v_{i,j}}{\lambda_{\min} (1+\eps)^j}  \kappa^\ell
+ [2 n (1+\eps) + 1] \kappa^\ell\,.
\end{equation}
We now derive a lower bound on the cost of the VMs used by \sha($\eps$) for handling jobs contributing to $\tilde{\Lambda}^v_{i,j}$, which serves as a lower bound on the number of VMs used by \sha($\eps$) for handling the total load $\tilde{\Lambda}^v_{i,j} + \bar{\Lambda}^v_{i,j}$.

First, note that the latency constraint of each top job contributing to $\tilde{\Lambda}^v_{i,j}$ is exactly
$\frac{1}{\bar{\mu}-\lambda_{\min}(1+\eps)^{j+1}}$.
Using similar arguments as those used to prove Lemma~\ref{lem:highshare:number_of_VMs}, we can conclude that the maximum load on any VM running any such job is $\lambda_{\min}(1+\eps)^{j+1}$, which we view as the {\em size} of the VMs used by \sha($\eps$) for handling these jobs.
Recall that \sha($\eps$) can place jobs fractionally, and it does not place
jobs with different latency range on the same VM.
It follows that all
VMs handling such jobs at layer $\ell$, except for at most one, are
completely full (recall that \sha($\eps$) is optimal and
therefore, without loss of generality, can be assumed to place all jobs at a level in a single node at that level). 
Then, all these full VMs  
have a load equal to their size. 
Consequently, the number of full VMs required by \sha($\eps$) for handling jobs contributing to $\tilde{\Lambda}^v_{i,j}$ is at least:
$\Bigl\lfloor\frac{\tilde{\Lambda}^v_{i,j}}{\lambda_{\min} (1+\eps)^{j+1}}\Bigr\rfloor \geq \frac{\tilde{\Lambda}^v_{i,j}}{\lambda_{\min} (1+\eps)^{j+1}} - 1,$
which translates into
\begin{equation}
\label{eq:sha_cost}
\tilde{\phi}(\msha^v_{i,j}) \geq \left( \frac{\tilde{\Lambda}^v_{i,j}}{\lambda_{\min} (1+\eps)^{j+1}} - 1 \right) \kappa^\ell,
\end{equation} 
where $\tilde{\phi}(\msha^v_{i,j})$ is the cost associated with running only the {\em full} VMs used by \sha($\eps$).
Hence, we have:
\begin{eqnarray}
\hspace{-5mm}\phi(\mhighshare^v_{i,j}) \hspace{-2mm}&\mathord{\leq}&\hspace{-2mm}
 2(1+\eps) \left(  \frac{\tilde{\Lambda}^v_{i,j}}{\lambda_{\min} (1+\eps)^{j+1}} - 1 \right) \kappa^\ell \nonumber\\
\hspace{-2mm}& &\hspace{-2mm} + [(2 n +2)(1+\eps) + 1] \kappa^\ell \label{eq:bound:phi:final_bound} \\
\hspace{-2mm}&\mathord{\leq}& \hspace{-2mm} 2(1
                              \mathord{+}\eps)\tilde{\phi}(\msha^v_{i,j})
                              \mathord{+} [(2 n \mathord{+}2)(1
                              \mathord{+}\eps) \mathord{+} 1] 
\kappa^\ell, \nonumber 
\end{eqnarray}
where the first inequality follows from (\ref{eq:highshare_cost}), and the second from (\ref{eq:sha_cost}).
By the relaxation property of \sha($\eps$),  the overall cost of VMs used by \sha($\eps$), and working at full computation capacity, serves as a lower bound on the cost of \opt.
Summing over all nodes $i$,  VNFs $v$, and latency ranges $j$, we get: 
\begin{align}
\label{eq:REShare_vs_OPT}
\phi(\mhighshare(\eps))
 &\leq \\
& \hspace{-15mm} 2(1 \mathord{+} \eps) \phi(\mopt) \mathord{+} ((2 n \mathord{+}2)(1\mathord{+}\eps) \mathord{+} 1) J_{\eps} \abs{V} \sum_i \kappa^\ell. \nonumber
\end{align}
Since $J_{\eps} \leq \log_{(1+\eps)}\frac{\bar{\mu}}{\lambda_{\min}}-1$,
Eq.~(\ref{eq:REShare_vs_OPT})  implies that
\[\lim_{\phi(\mopt) \to \infty} \frac{\phi(\mhighshare(\eps))}{\phi(\mopt)} \leq 2(1+\eps)\,,\] 
thus completing the proof.
\end{proof}

\section{Dynamically adjusting $\eps$: \dhighshare}
\label{sec:DHS}

In this section, we present the \dhighshare\ algorithm, which dynamically adjusts $\eps$ in order to 
optimize the deployment of service instances  
as the service request load varies arbitrarily over time.
The design of \dhighshare\ draws upon the analysis of the competitive ratio of
\highshare($\eps$). 
Our analysis indicates that we get a better {\em multiplicative} competitive ratio when decreasing $\eps$, implying that, ideally, for infinite duration requests, and ever increasing load, it is best to set $\eps$ as small as possible.
Our analysis also shows that the {\em additive} terms in the competitive performance of \highshare($\eps$) increase as   $\eps$ decreases.
Thus, using very small values of $\eps$ may end up being inefficient when the load is small.
Consequently,  \dhighshare\ starts with a large $\eps$ to minimize the constant overheads, and, as the load increases, it reduces $\eps$ to optimize the asymptotic competitive ratio. Finally, since \dhighshare\ explicitly deals with time-varying workloads, it increases $\eps$ when the total load declines, as the constant terms have a larger impact on its performance.

To do so, \dhighshare\ simulates our shadow strategy \sha($\eps$) alongside the actual decisions it performs and uses the cost of \sha($\eps$) as an indicator of the system load, which is then leveraged to dynamically adjust the value of $\eps$ being used by \dhighshare.
Specifically, as the service request load increases, \highshare($\eps$) gets asymptotically closer to \sha($\eps$) when $\eps$ is small, at the cost of a larger additive term due to the VMs partitioning into latency ranges $L_j$  ($0\leq j\leq J_{\eps}$).
This notion yields the main design criterion for \dhighshare: $\eps$ should be
gradually reduced as the load grows, and instead increased when
the load drops, as also depicted in \Fig{PoD}. 
Importantly, we show that
\dhighshare\ has low complexity and
it is asymptotically 2-competitive when requests have infinite duration, i.e., $\tau_r=\infty$;
also, as shown in \Sec{results}, it can effectively cope with very diverse, practical scenarios.

\dhighshare\ can be seen as a way to apply different versions of \highshare($\eps$), constantly adjusting the value of~$\eps$ as the service load varies over time.
Specifically, \dhighshare\ begins by running \highshare($\eps$)
with a large initial value of $\eps$;
at the same time, it simulates \sha($\eps$) and keeps track of its cost.
When the load goes beyond a threshold that, as detailed below, depends on the ratio between \highshare($\eps$)'s and \sha($\eps$)'s costs, we keep track of the current load and reduce $\eps$. This action is performed again and again as long as the load increases,
following the arrival of new requests.

Once the load drops below the previous load mark,
which is essentially due to requests leaving the system,
we increase $\eps$ to its previous value. This approach is applied repeatedly as long as the load decreases.

Finally, we remark that, although we consider that service instances arrive and leave, their traffic load remains constant during their lifetime.
\dhighshare\ can also cope with time-varying values of $\lambda_r$. In particular, one can look at a change in the value of $\lambda_r$ as if the current service instance left and another, exhibiting the new value of $\lambda_r$, arrived.  

\subsection{Algorithm description}\label{sec:algo-D}

To formally define \dhighshare, we use the following notation.
Let $t_0$ be the arrival time of the first request. We then look at
sequence $\sigma$ as the concatenation of subsequences $\sigma_1,
\sigma_2, \ldots$, such that $\sigma_{\RSphase}$ is the sequence of requests
arriving/departing in interval $I_{\RSphase}=[t_{{\RSphase}-1},t_{\RSphase})$, for $\RSphase \geq 1$.
Intervals are periods of time during which $\eps$ remains
unchanged: 
for every $I_{\RSphase}$, \highshare($\eps$)  uses a given $\eps$ to determine the latency ranges, as described in \Alg{highshare}.
We indicate such value by $\eps_{h(\RSphase)}$, where $h(\RSphase)$ denotes the index of the level of the system load at the beginning of  $I_{\RSphase}$.
The algorithm keeps track of the most recent load threshold, associated with index $h(\RSphase)$, throughout its execution, using parameter $\dload_{h({\RSphase})}$.
Also, let us define:  
\begin{align}
Y_{\RSphase} &= \tilde{\phi}(\msha(\eps_{h(\RSphase)}),\sigma_{\RSphase}) \,,
\label{eq:Y_n}
\\
Z &= [(2n+2) (1+\eps^*) + 1] \log \frac{\bar{\mu}}{\lambda_{\min}}
    \abs{V} \sum_i \kappa^\ell \,,
\label{eq:Z}
\end{align}
where, similarly to the proof of Theorem~\ref{thm:REShare_competitive_ratio}, $\tilde{\phi}(\msha(\eps_{h(\RSphase)}),\sigma_{\RSphase})$ is the cost associated with running only the full VMs used by $\msha(\eps_{h(\RSphase)})$.
Here, $\eps^*>0$ is an initial value that satisfies $\eps^* \geq
\eps_{h(\RSphase)}$, $\forall \RSphase$. 
We remark that the cost in \Eq{Y_n} refers to the end of the interval $I_{\RSphase}$ and that  
 $Z$ represents the second term in
 Eq.~(\ref{eq:bound:phi:final_bound}), which is independent of both $\RSphase$ and $h(\RSphase)$.
For every request $\rb$ arriving or departing during $I_{\RSphase}$, we let $\sigma_{\RSphase}^{\rb}$ denote the subsequence of request arrivals and departures in $\sigma_{\RSphase}$ up to (and including) the arrival/departure of $\rb$.
We further extend the above notation and define $Y_{\RSphase}^{\rb} = \tilde{\phi}(\msha(\eps_{h(\RSphase)}), \sigma_{\RSphase}^{\rb})$.
At any time $t$, let $\tilde{Y}_{\RStildephase}$ be the value of $Y_{\RSphase}$ for the last interval in which $\eps=\eps_{\RStildephase}$, over all intervals $I_\RSphase$ ending no later than $t$.
We then define:
\begin{align}
C_{\RSphase} &= \frac{Z}{\eps_{h(\RSphase)} \log (1+\eps_{h(\RSphase)})}\,,
\label{eq:C_n}
\\
S_{\RSphase} &= \frac{1}{\eps_{h(\RSphase)}} \sum_{\RStildephase=1}^{h(\RSphase)-1} (2 + 3  \eps_{h(\RStildephase)}) \tilde{Y}_{\RStildephase}\,.
\label{eq:S_n}
\end{align}
Note that during interval $I_{\RSphase}$, both $C_{\RSphase}$ and $S_{\RSphase}$  are fixed.

\dhighshare\ is formally defined in \Alg{dynamichighshare}. 
\dhighshare\ takes as input the value $\eps^*$, which is the initial value of $\eps$ and is used to define $Z$. 
For each $\RSphase=1,2,\ldots$, the subsequence $\sigma_{\RSphase}$ will be implicitly defined during the execution of \dhighshare, according to the requests $\rb$ handled by the algorithm between consecutive updates to the value of $\eps$.
Index $h(\RSphase)$, instead, keeps track of the evolution of the system load
(by updating $\dload_{h(\RSphase)+1}=\Lambda$ in
line~\ref{dhighshare:update_load_threshold}, with $\Lambda$ being the
current system-wide load), and of the corresponding value of $\eps$
given as 
input to \highshare, which is used as a subroutine within \dhighshare\
(\highshare($\eps_{h(\RSphase)}$)).
It follows that, by calling \highshare\ with different values of the parameter~$\eps_{h(\RSphase)}$, \dhighshare\ can adjust to the traffic load as this changes over time, always using the most appropriate values of~$\eps$.

\begin{algorithm}[t]
\caption[\dhighshare]{\dhighshare}
\label{alg:dynamichighshare}
\begin{algorithmic}[1]
\State init $\RSphase=1$; $h(1)=1$; $\eps_{h(1)} = \eps^*$; $\dload_{h(1)}=0$
\For{{\bf each} service request $\rb$ arriving or departing}
    \State {\bf handle} $\rb$ in \highshare($\eps_{h(\RSphase)})$
    \State {\bf update} the cost $\rb$ in \sha($\eps_{h(\RSphase)})$
    \If{$Y_{\RSphase}^r \geq \max\set{C_{\RSphase}, S_{\RSphase}}$}
        \Comment{arrival, cost (load) above threshold}
        \label{dhighshare:increase_condition}
        \State $\eps_{h(\RSphase)+1} \gets \eps_{h(\RSphase)} / 2$
        \label{dhighshare:decrease_eps}
        \Comment{reduce $\eps$}
        \State $\dload_{h(\RSphase)+1} \gets \Lambda$
        \label{dhighshare:update_load_threshold}
        \Comment{set load threshold}
        \State $\tilde{Y}_{h(\RSphase)} \gets Y_{\RSphase}^r$
        \label{dhighshare:Y_h_update}
        \State $h(\RSphase+1) \gets h(\RSphase)+1$
        \Comment{update load level index}
       \State $\RSphase \gets \RSphase+1$
    \ElsIf{$\Lambda < \dload_{h(\RSphase)}$}
        \label{dhighshare:load_below_threshold}
        \Comment{departure, load below threshold}
        \State $h(\RSphase+1) \gets h(\RSphase)-1$
        \Comment{update load level index}
\State $\RSphase \gets \RSphase+1$
    \EndIf
\EndFor
\end{algorithmic}
\end{algorithm}

Assume that the condition in Line~\ref{dhighshare:increase_condition} holds for some $\RSphase$ and request $\rb$ arriving or departing during $\sigma_{\RSphase}$.
Since the right-hand side of the condition is fixed, this implies that
$Y_{\RSphase}^r$ has increased due to handling $r$, which means that $r$ has arrived at the system, causing $Y_{\RSphase}^r$ to increase beyond the value of the right-hand side.
Since $Y_{\RSphase}^r \geq C_{\RSphase}$, along with the fact that in Line~\ref{dhighshare:Y_h_update} $\tilde{Y}_{h(\RSphase)}$ is set to be $Y_{\RSphase}^r$, we have:
\begin{equation}
\label{eq:Y-Z}
\eps_{h(\RSphase)}  \tilde{Y}_{h(\RSphase)} \geq Z  \frac{1}{\log (1+ \eps_{h(\RSphase)})}.
\end{equation}
Intuitively, using the insight derived from the analysis presented in Sec.~\ref{sec:HS}, Eq.~(\ref{eq:Y-Z}) implies that a mere $\eps_{h(\RSphase)}$ fraction of the cost $\tilde{\phi}(\msha(\eps_{h(\RSphase)}),\sigma_{\RSphase})$ is already sufficiently larger than the cost incurred by \dhighshare\ for handling requests in different nodes than the ones used by \sha($\eps_{h(\RSphase)}$), for input sequence $\sigma_{\RSphase}$.
Also, since 
$Y_{\RSphase}^r \geq S_{\RSphase}$, it follows that:
\begin{equation}
\label{eq:Y-sum_Y_n}
\eps_{h(\RSphase)}  \tilde{Y}_{h(\RSphase)} \geq \sum_{\RStildephase=1}^{h(\RSphase)-1} (2 + 3 \eps_{\RStildephase}) \tilde{Y}_{\RStildephase}.
\end{equation}
Eq.~(\ref{eq:Y-sum_Y_n}) means that a fraction $\eps_{h(\RSphase)}$ of
$\tilde{Y}_{h(\RSphase)}$, i.e., a fraction $\eps_{h(\RSphase)}$ of the cost $\tilde{\phi}(\msha(\eps_{h(\RSphase)}),\sigma_{\RSphase})$, is
already sufficiently larger than the cumulative cost of  \sha\ over
the previous time intervals to warrant a change of $\eps$.
We note that these lower bounds on the {\em cost} of
\sha($\eps_{h(\RSphase)})$ are commensurate to the {\em load} encountered by
the system (both by \sha($\eps_{h(\RSphase)}$) and \dhighshare), during
interval $I_\RSphase$, and they come in handy for the analysis of
\dhighshare's competitive ratio. 

When service requests expire and make the overall system load  drop below the previous threshold (Line~\ref{dhighshare:load_below_threshold}), the algorithm reverts to using the previous load level index ($h(\RSphase)-1$), and, consequently, its corresponding value of $\eps$.

\subsection{Complexity and competitive ratio analysis}
\label{sec:competitive_analysis_REShare}

We first observe that the complexity of \dhighshare\ is remarkably low: from inspection of
Algs.\,\ref{alg:highshare}--\ref{alg:dynamichighshare} and considering
the complexity of \sha($\eps$), the total 
complexity of  \dhighshare\ is $O(B|\Vc_r^s|)$ 
where $B$ is the number of VMs currently used.

In the remainder of this subsection, we assume that requests have unlimited
$\tau_r$, i.e., there are no departures, and we analyze the competitive
ratio 
of \dhighshare\ in such a setting.
In particular, under this assumption, we have that $h(\RSphase)=\RSphase$ for all $\RSphase$.
For simplicity, we henceforth use index $h$ to represent both $\RSphase$ and $h(\RSphase)$.
Furthermore, under these settings $\tilde{Y}_h=Y_h$, thus we simply use $Y_h$ to denote either $\tilde{Y}_h$ or $Y_h$.

To prove \dhighshare's competitive ratio, we proceed as follows. 
First, let us recall the observations made in \Sec{algo-D}, i.e., whenever the condition in Line~\ref{dhighshare:increase_condition} of \Alg{dynamichighshare} holds, the overall load handled by the system at that point is significantly higher than the load at the beginning of the interval.
In such a case, we reduce the value of $\eps$ to be used in the subsequent interval (Line~\ref{dhighshare:decrease_eps}).
Since  index $h$ serves as a proxy to the overall load in the system, this update rule for the value of $\eps$  implies that $\lim_{h \to \infty} \sum_{\RStildephase=1}^h \phi(\msha(\eps_{\RStildephase},\sigma_{\RStildephase})) = \infty$ iff $\lim_{h \to \infty} \eps_h=0$. This is due to the fact that ever-decreasing values of $\eps$ yield ever-narrower latency ranges, hence a higher number of used VMs. 
Second, 
\highshare($\eps_h$), as well as \sha($\eps_h$),  handle the requests arriving in different intervals using a new series of latency ranges, hence they cannot reuse already deployed VNFs. Thus,  the costs associated with different intervals $I_h$ can be considered separately. 
Third, we prove the following lemma, which is the key result to
derive the asymptotic competitive ratio of \dhighshare. 
\begin{lemma}
\label{lem:dynamic_RS}
For every $k=1,2,\ldots,$
\begin{equation}
\label{eq:dynamic_eps_guarantee}
\sum_{h=1}^k \phi(\mdhighshare, \sigma_h) \leq (2+4\eps_k) \sum_{h=1}^k\tilde{\phi}(\msha(\eps_h), \sigma_h). \nonumber
\end{equation}
\end{lemma}
\begin{proof}
By the definitions of $Y_h$ and $Z$ given in Eqs.~(\ref{eq:Y_n})-(\ref{eq:Z}),  Eq.~(\ref{eq:bound:phi:final_bound}) implies (after changing the logarithm base) that
\begin{align}
\phi(\mdhighshare, \sigma_h) \leq  (2+2 \eps_h)  Y_h + Z \frac{1}{\log (1+\eps_h)}.
\label{eq:RS_h_vs_Y_h_plus_Z}
\end{align}

By plugging Eq.~(\ref{eq:Y-Z}) into Eq.~(\ref{eq:RS_h_vs_Y_h_plus_Z}), we obtain
\begin{align}
\phi(\mdhighshare, \sigma_h) \leq  (2+3 \eps_h)  Y_h.
\label{eq:RS_h_vs_Y_h}
\end{align} 
By summing over $h$ in \Eq{RS_h_vs_Y_h} and using \Eq{Y-sum_Y_n}, 
we get 
\begin{eqnarray}
\sum_{h=1}^{k} \phi(\mdhighshare, \sigma_h)
&\leq &
(2 + 3 \eps_k)Y_k + \sum_{h=1}^{k-1} (2 + 3 \eps_h) Y_h \nonumber\\
&\leq&
(2+ 4 \eps_k) Y_k.
\label{eq:sum_k_RS_upper_bound}
\end{eqnarray}
By the definition of $Y_k$, the thesis follows.
\end{proof}

The following theorem is an immediate corollary of
Lemma~\ref{lem:dynamic_RS}, for the case of unbounded $\tau_r$.
\begin{theorem}
\label{thm:dynamic_RS_2_competitive}
When $\tau_r$ is unbounded, \dhighshare\ is asymptotically 2-competitive. 
\end{theorem}
\begin{proof}
By the update rule of $\eps_h$,
in line~\ref{dhighshare:decrease_eps} of \dhighshare,
we have that $\lim_{h\to\infty}\eps_h = 0$.
By using Lemma~\ref{lem:dynamic_RS}, this therefore implies that
$\lim_{h \to \infty} \frac{\phi(\mdhighshare)}{\sum_{\RStildephase=1}^h\tilde{\phi}(\msha(\eps_{\RStildephase},\sigma_{\RStildephase}))} \leq 2$.
Since
$\sum_{\RStildephase=1}^h\tilde{\phi}(\msha(\eps_{\RStildephase},\sigma_{\RStildephase}))$ is a lower bound on $\phi(\mopt)$,
the theorem follows.
\end{proof}

\section{Numerical results}
\label{sec:results}

In this section, we describe the scenarios we use for our performance evaluation (\Sec{sub-scenarios}), the workloads we consider (\Sec{sub-services}), the benchmark solutions we compare against (\Sec{sub-benchmarks}), and the results we obtain (\Sec{sub-results}).

\subsection{Reference scenarios}
\label{sec:sub-scenarios}

We consider two hierarchical topologies akin those originally proposed
in~\cite{tong2016hierarchical} and used in many research works
thereafter. The first is organized
 in three layers \cite{martin2019modeling}:
(i) {\em edge}, closest to the users (leaf nodes) but the most expensive (normalized, per-VM fixed cost of 7.5);
(ii) {\em aggregation}, cheaper than edge (normalized, per-VM fixed cost of 2.5) but incurring a moderate extra traffic forwarding  latency; 
(iii) {\em cloud}, cheapest (normalized, per-VM fixed cost of 1) but associated with the longest extra latency. 
The second is a four-layer topology, including a {\em fog} layer that incurs a per-VM fixed cost of 10.
Cost figures are obtained from~\cite{5gt-d14}, presenting the technical and economic aspects of datacenters of different sizes, which can be placed at different network segments.

Each VM has a computing capability of $\bar{\mu}=100$ packets/ms, and a per-packet proportional cost of 1/100th its fixed cost.
This
reflects findings reported in~\cite{morabito2015power}, i.e., that idle VMs consume roughly half as much power as fully-utilized ones (with the latter incurring both fixed and proportional costs).
In the three-layer scenario, the extra latency associated with the aggregation and the cloud layer is, respectively, 15~ms and 30~ms. In the four-layer one, the fog layer has no extra latency, while the latency of all other layers is increased by 5~ms.
\subsection{Workloads}
\label{sec:sub-services}

We consider different workloads for our performance evaluation from three of the main application realms of slicing-enabled networks: connected vehicles, smart factories, and cloud-edge computing.

{\bf Vehicular domain.}
We begin by considering the three-layer topology and the main services of the vehicular domain, 
presented in~\cite{pimrc-wp3} and described in \Tab{services}:
\begin{itemize}
    \item Intersection Collision Avoidance (ICA): vehicles periodically broadcast  Cooperative Awareness Messages (CAMs)  including their position, speed, and acceleration; a collision detector checks if any pair of vehicles are on a collision course and, if so, it issues an alert;
    \item Vehicular see-through (CT): cars display on their on-board screen the video captured by the preceding vehicle, e.g., a large truck obstructing the view;
    \item Entertainment (EN): passengers consume streaming content, provided with the assistance of a content delivery network (CDN) server.
\end{itemize}
In addition to their service-specific VMs, all services leverage a
virtual communication sub-slice, 
as  listed at the top of \Tab{services}.
New service requests are generated whenever a new instance is needed,
e.g., the service has to be deployed at a new location, or the
crossing surveilled by ICA becomes more crowded, thus triggering a
service scale out.

\begin{table}
\caption{Services target delay and VNF complexity\label{tab:services}} 
\scriptsize{
\begin{tabularx}{\columnwidth}{|X|r||X|r|}
\hline
{\bf VNF} & {\bf $\theta_v$} & {\bf VNF} & {\bf  $\theta_v$} \\
\hline\hline
\multicolumn{2}{|>{\hsize=2\hsize}c||}{{\em Virtual comm. sub-slice}} 
& \multicolumn{2}{>{\hsize=2\hsize}c|}{{\em CT -- delay target 50\,ms}}\\
\hline
eNB & $1$ &
Car information management (CIM) & $10$\\
\hline
EPC PGW & $1$ &
CT server & $8$ \\
\hline
EPC SGW & $1$ &
CT database & $1$\\
\hline
EPC HSS & $1$ & \multicolumn{2}{>{\hsize=\columnwidth}c|}{{\em EN -- delay target 1\,s}}\\
\hline
EPC MME & $10$ & Video origin server & $10$  \\
\hline
\multicolumn{2}{|>{\hsize=\columnwidth}c||}{}  & Video CDN & $3$\\
\hline
\multicolumn{2}{|>{\hsize=\columnwidth}c||}{{\em ICA -- delay target 20\,ms}} & 
\multicolumn{2}{>{\hsize=\columnwidth}c|}{{\em Smart-factory -- delay target 100\,ms}}\\
\hline
Car information management (CIM) & $7$ &
Robot controller & $10$\\
\hline 
Collision detector & $10$ 
& Motion planning & $10$\\
\hline
Car manufacturer database & $1$ 
& Configuration interface & $5$\\
\hline
Alarm generator & $1$ &
Digital twin application & $10$\\
\hline
\end{tabularx}
} 
\end{table}


The computational requirements of the VNFs (i.e., their complexity) reported in \Tab{services} come from~\cite{pimrc-wp3}; instead, the load we apply is synthetically generated, in order to demonstrate how our approach handles rapid demand fluctuations. Specifically,
the service request process is  as follows:
\begin{enumerate}
    \item for the first 15\,s of the time horizon under study, a new request arrives every second;
    \item after that, and until 800\,s from the beginning of the horizon, no further requests arrive;
    \item between 800 and 1,000\,s from the beginning, 1,000 more requests (5 per second) arrive;
    \item those requests leave the system, at the same rate,
      between 1,000 and 1,200\,s  after the beginning of the horizon.
\end{enumerate}
The first requests represent long-running services, which are active even during periods of low traffic. The subsequent request arrivals  represent a sudden surge in vehicular activity, to which more service requests are associated, and an equally sudden decrease thereof.

{\bf Smart-factory domain.} 
Digital twins are computer models of real objects, controlling the behavior of their physical counterparts. As detailed in~\cite{digitaltwin}, in smart-factory scenarios semi-autonomous robots are controlled by entities running within the network infrastructure. Also in this case  we consider the three-layer scenario,  along with the services specified in \Tab{services}. The
 main tasks to perform are:
(i)
 fine-grained {\em control} of robot actions;
 (ii) {\em planning} of their actions and mobility;
 (iii) {\em configuration} of the robots;
 (iv) the {\em digital twin} itself. 
Comparing the smart factory to the vehicular service characteristics in \Tab{services}, one can notice how, while the end-to-end delays required in the two scenarios are comparable, the structure of the services is fairly different. Specifically, the smart-factory scenario only includes one service, including four VNFs, besides the virtual communication sub-slice.
It follows that comparing the performance of \dhighshare\ across these two use cases captures both quantitative differences and qualitative variability.

{\bf Real-world computing load.}
Finally, we move to the four-layer scenario and consider a real-world scenario where the demand comes from the GWA Materna  trace~\cite{kohne2014federatedcloudsim,kohne2016evaluation}, which depicts the real-world evolution of the demand of a major cloud operator in Europe.

We also consider a different service request arrival process, where requests arrive and depart faster. Specifically:
\begin{enumerate}
    \item for the first 15\,s of the time horizon under study, a new request arrives every second;
    \item after that, and until 780\,s from the beginning of the horizon, no further requests arrive;
    \item between 780 and 800\,s from the beginning, 1,000 more requests (50 per second) arrive;
    \item those requests leave the system, at the same rate,
      between 1,180 and 1,200\,s  after the beginning of the horizon.
\end{enumerate}

\subsection{Benchmark strategies}
\label{sec:sub-benchmarks}

We compare \dhighshare\ against two benchmarks.
The first one is the shadow assignment used in \Sec{compratio-HS}.
As discussed in \Sec{compratio-HS}, the cost of full VMs used by such a strategy is a lower bound on the optimum because it is allowed to place different parts of the same request at different hosts, which neither the optimum nor \dhighshare\ are, of course, allowed to do, and it also handles relaxed delay constraints.
Notice how \sha\ is not a concrete strategy that could be applied in a real-world scenario, but rather it serves as a proxy of the lower bound on the optimum cost. In other words, the closer to \sha\ a strategy is, the better that strategy performs.

The second benchmark, labeled {\em RelaxSoTA} in the plots, is an adaptation
to our scenario of
the highly influential works~\cite{tulino,tulino2};
importantly, the same approach has been later followed by further studies on edge computing~\cite{jovsilo2022joint}, distributed machine learning~\cite{dinh2022network}, and energy-efficient mobile gaming~\cite{spinelli2022migration}.
Such approaches are based upon solving a convex relaxation of the placement problem every time a new request arrives, and making the placement decisions associated with the highest values of the relaxed variables.
Notice that, unlike \dhighshare, the RelaxSoTA approach may place VNFs of the same service at different layers.

In summary, we can say that both benchmarks have an advantage over \dhighshare, respectively, the ability to split VNFs and to spread a service across multiple layers. Therefore, comparing \dhighshare\ against such powerful alternatives yields additional relevance to our results.

{\bf Price-of-Dissimilarity (PoD).} The PoD captures the additional computational
capacity used because of the difference in delay constraints among the jobs served by the same VM (see also \Sec{RM}).
By the definition of the latency incurred on VM $b$ running VNF $v$ at speed $\mu_b$,  the delay incurred by {\em any} job assigned to $b$ is
$\frac{1}{\mu_b-\theta_v \Lambda(b)}$.
We then define the PoD of VM $b$ running on node $i$ as:
$
\max_r\set{ \frac{1}{D_r^v} - \left (\mu_b-\theta_v \Lambda(b) \right ) \mid (r,v) \mbox{ is assigned to } b}
$.
For any feasible assignment, the PoD of any VM $b$ is non-negative; also, 
when {\em all} jobs assigned to $b$ have the same delay constraint (as in \sha), it is possible to pick $\mu_b$ such that the PoD is zero.

\begin{figure*}
\centering
\includegraphics[width=.46\textwidth]{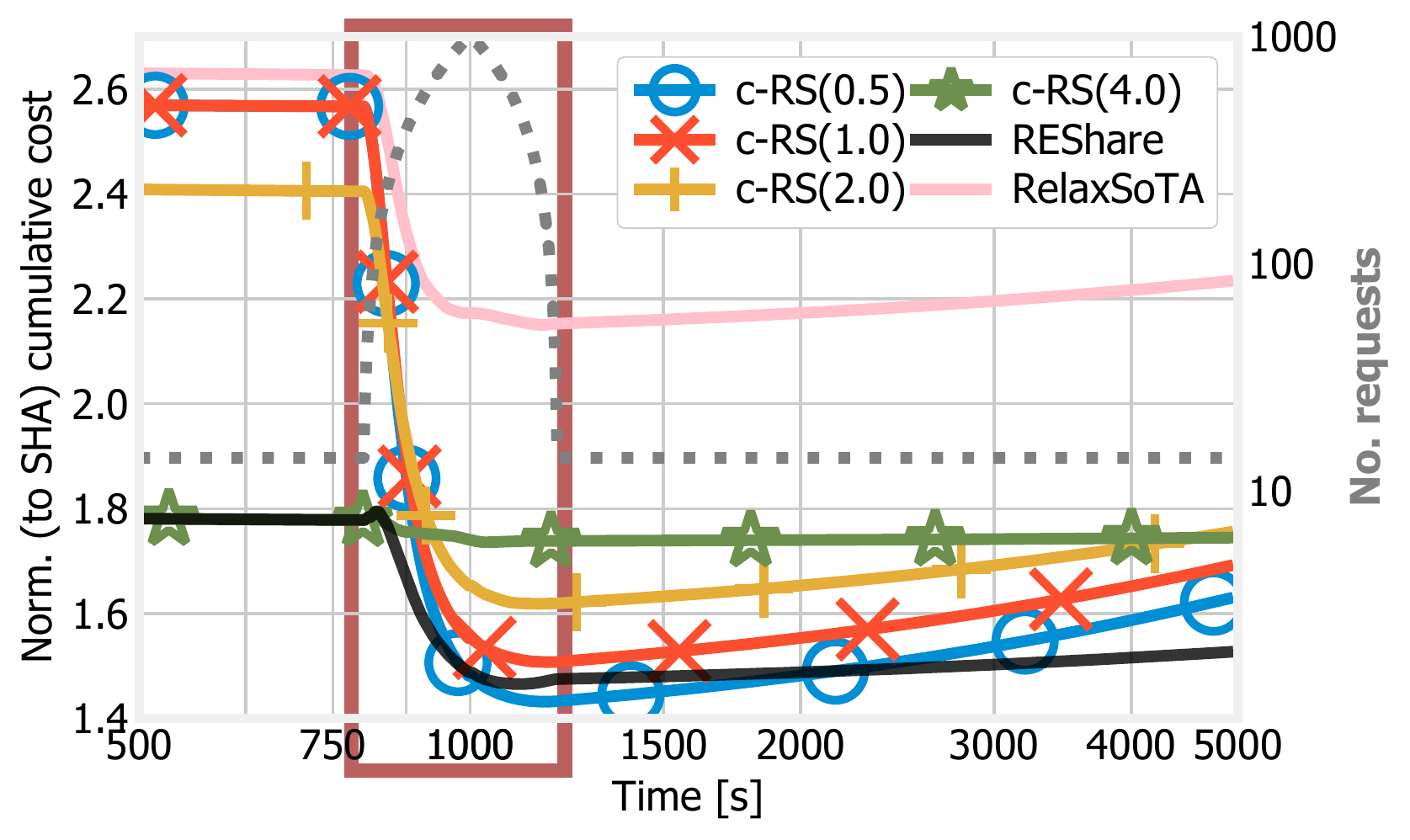}
\hspace{5mm}
\includegraphics[width=.46\textwidth]{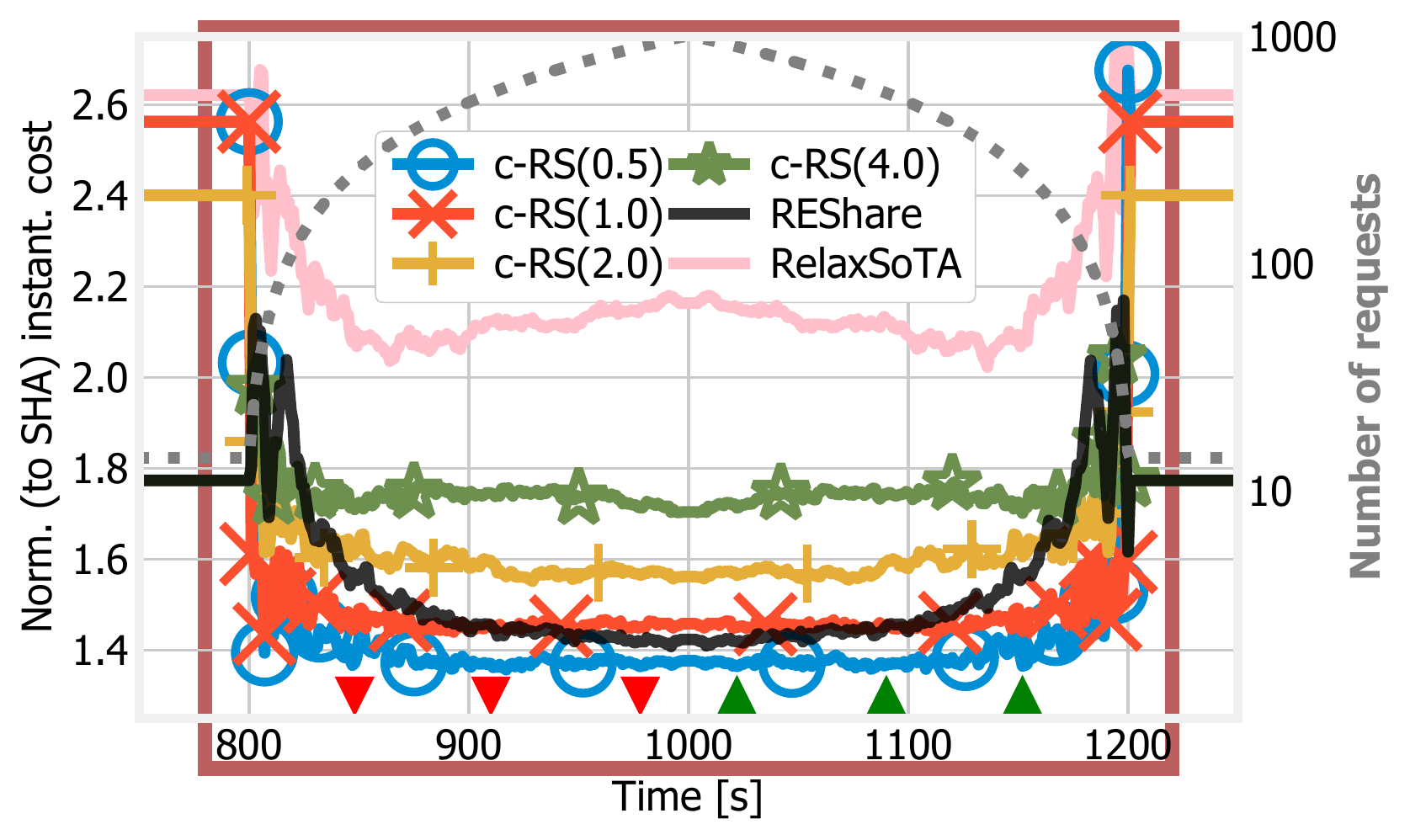}
\caption{
Three-layer scenario, vehicular application. \dhighshare\ and benchmark strategies: cumulative cost (left) and details of instantaneous cost during the load peak (right). In both plots, the dotted line corresponds to the load. In the right plot, upwards and downwards triangles at the bottom correspond to increasing and decreasing $\eps$ in \dhighshare.
Brown boxes denote the time period during which short-lived requests arrive and leave.
\label{fig:cost}
} 
\centering
\includegraphics[width=.32\textwidth]{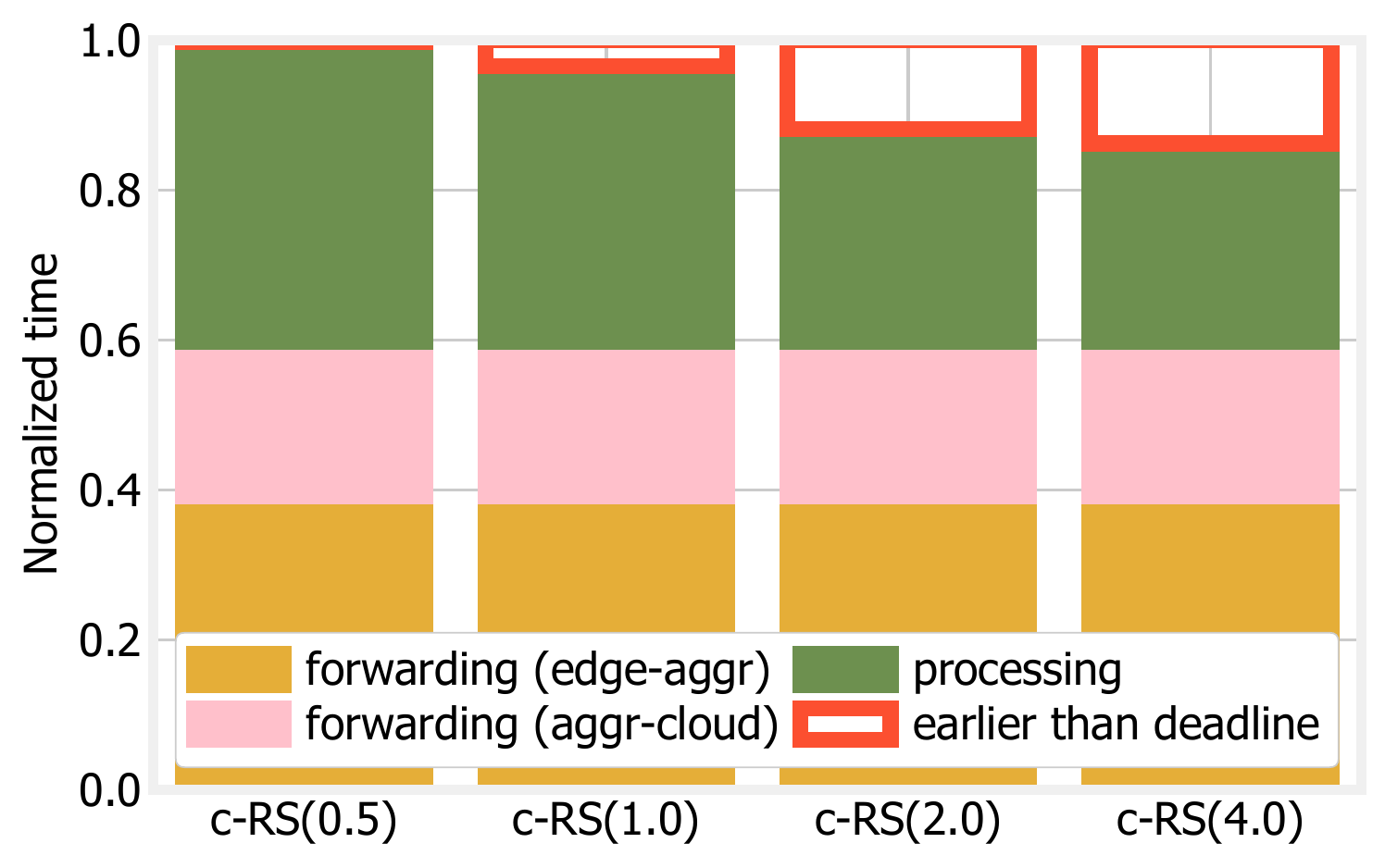}
\includegraphics[width=.32\textwidth]{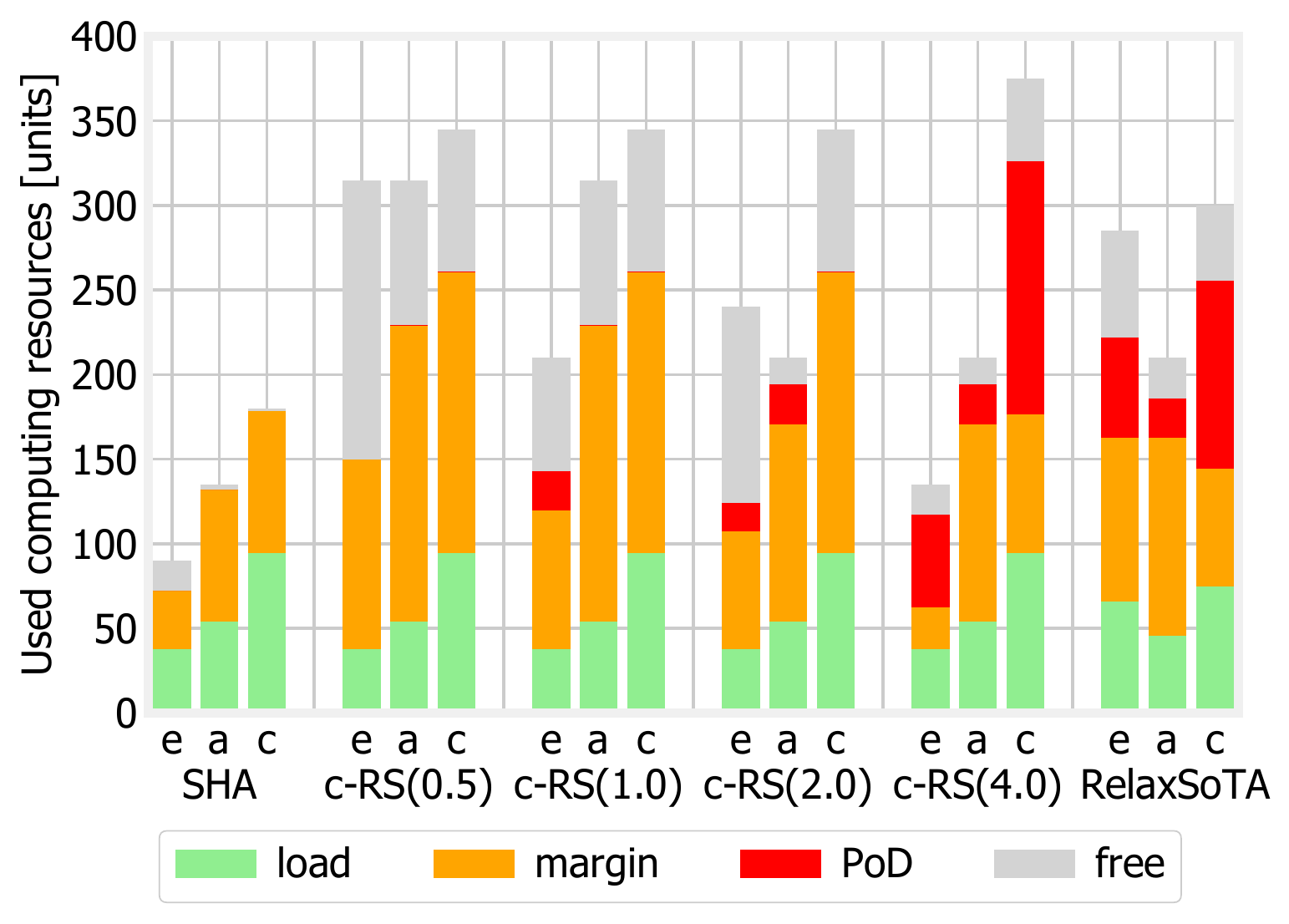}
\includegraphics[width=.32\textwidth]{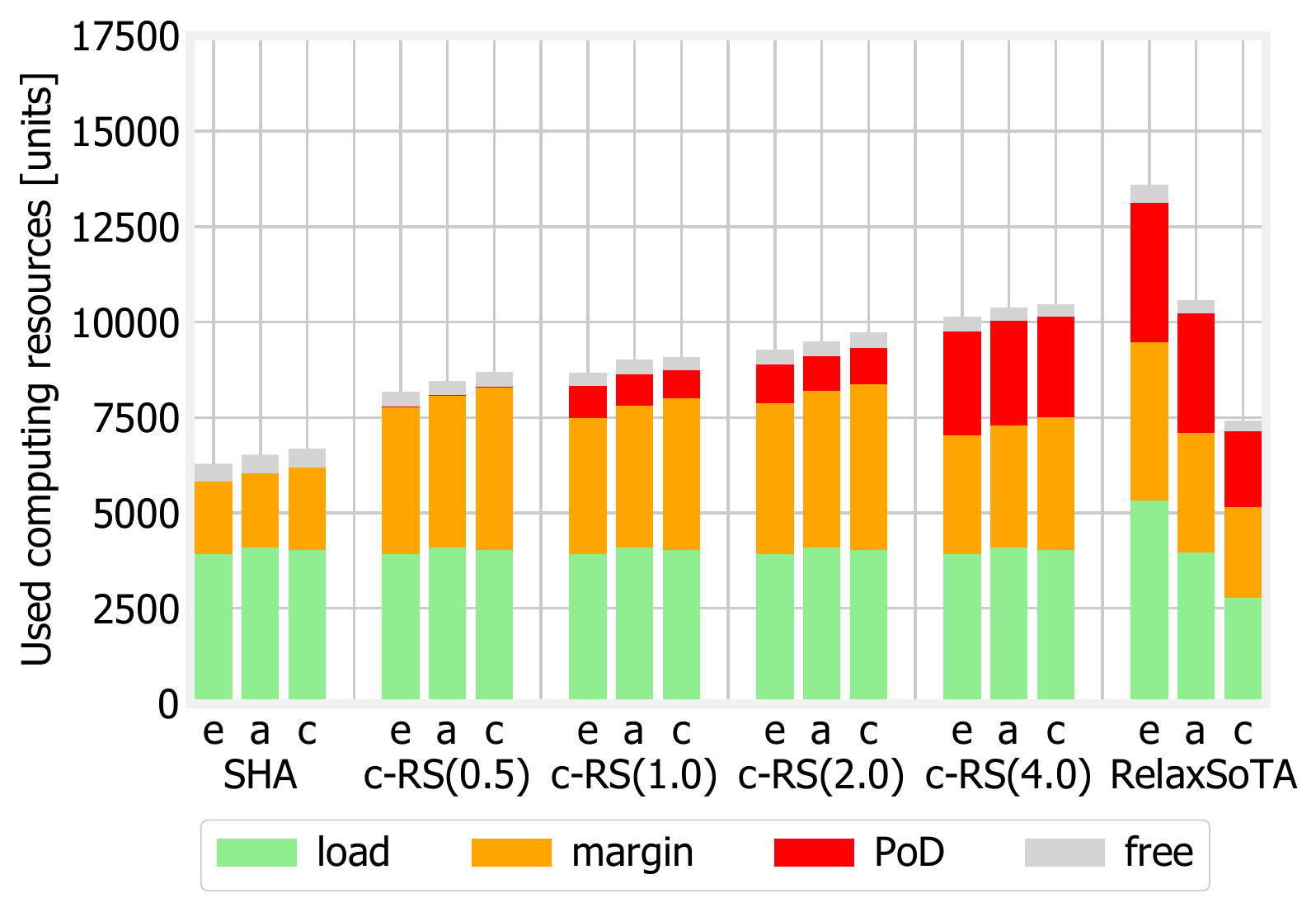}
\caption{
Three-layer scenario, vehicular application. \highshare($\eps$), with different values of $\eps$ (labeled by c-RS($\eps$) for short): actual service
latency normalized to the target value (left);  usage of computing
resources for when traffic is low, namely, after the first 15 (long-running) requests arrive (center), and at peak load (right).
\label{fig:details}
} 
\end{figure*}

\begin{figure*}
\centering
\includegraphics[width=.46\textwidth]{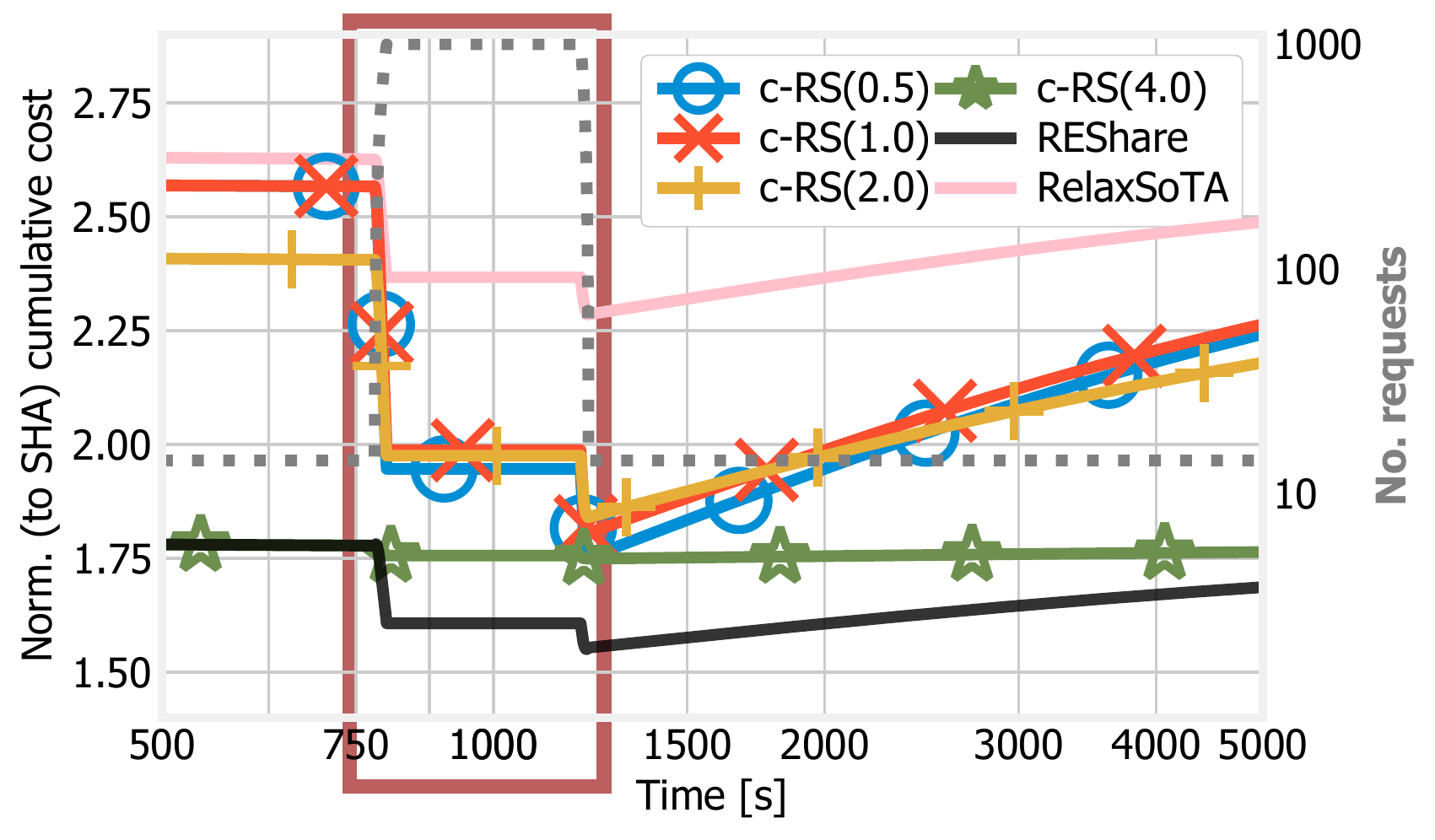}
\hspace{5mm}
\includegraphics[width=.46\textwidth]{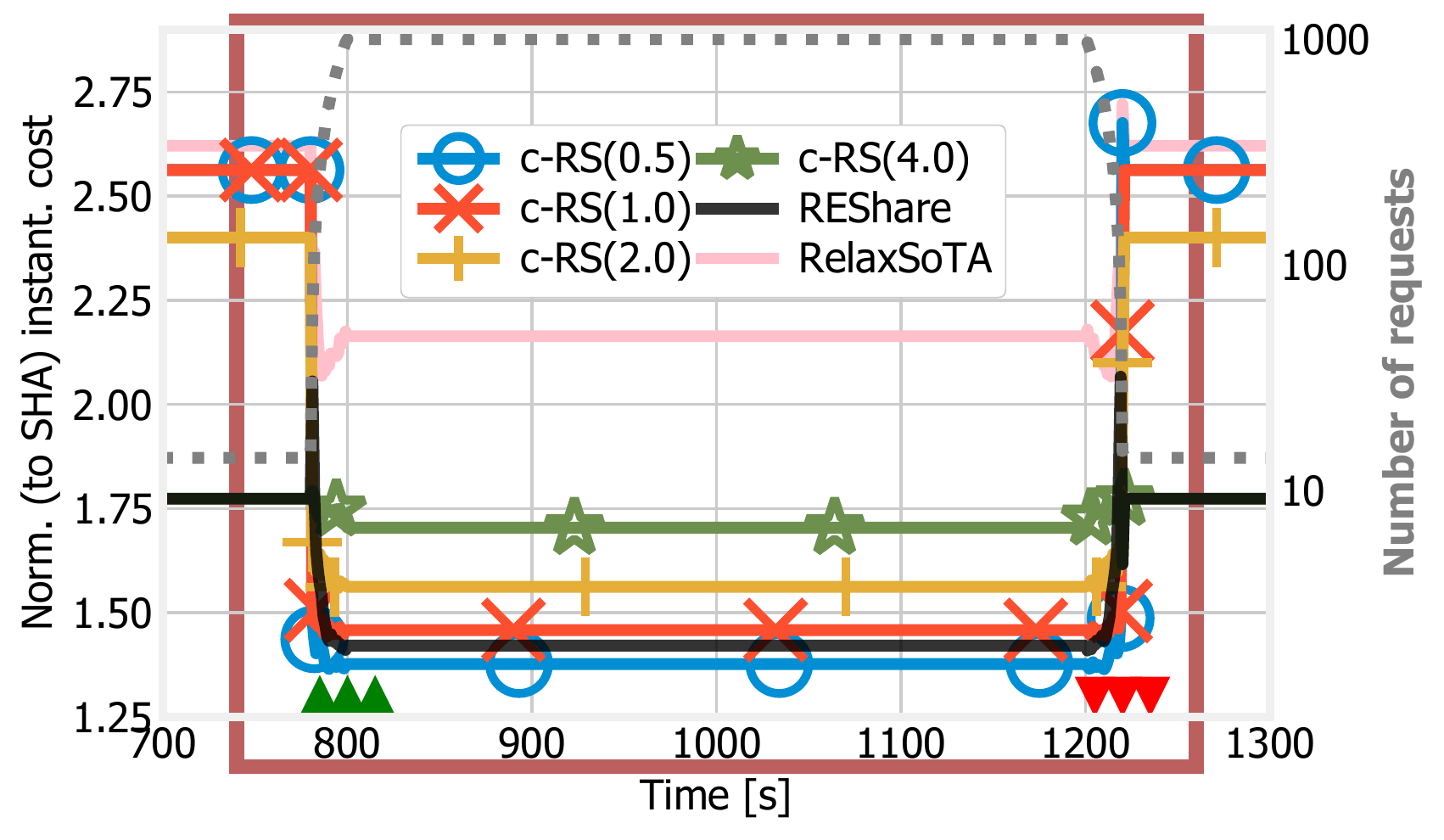}
\caption{
Three-layer scenario, smart-factory application. \dhighshare\ and benchmark strategies: cumulative cost (left) and details of instantaneous cost during the load peak (right). In both plots, the dotted line corresponds to the load. In the right plot, upwards and downwards triangles at the bottom correspond to increasing and decreasing $\eps$ in \dhighshare.
\label{fig:sf-cost}
Brown boxes denote the time period during which short-lived requests arrive and leave.
} 
\centering
\includegraphics[width=.32\textwidth]{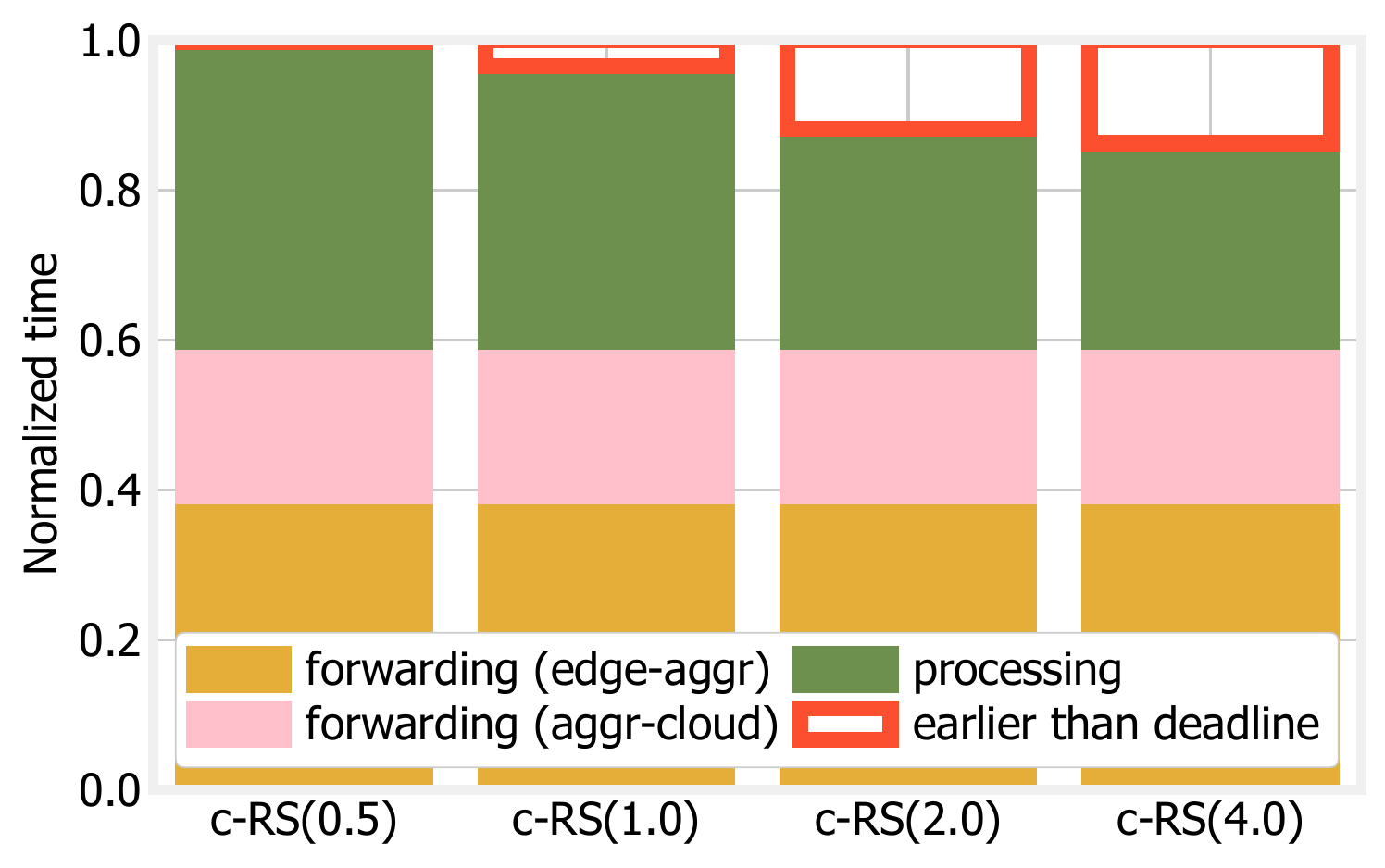}
\includegraphics[width=.32\textwidth]{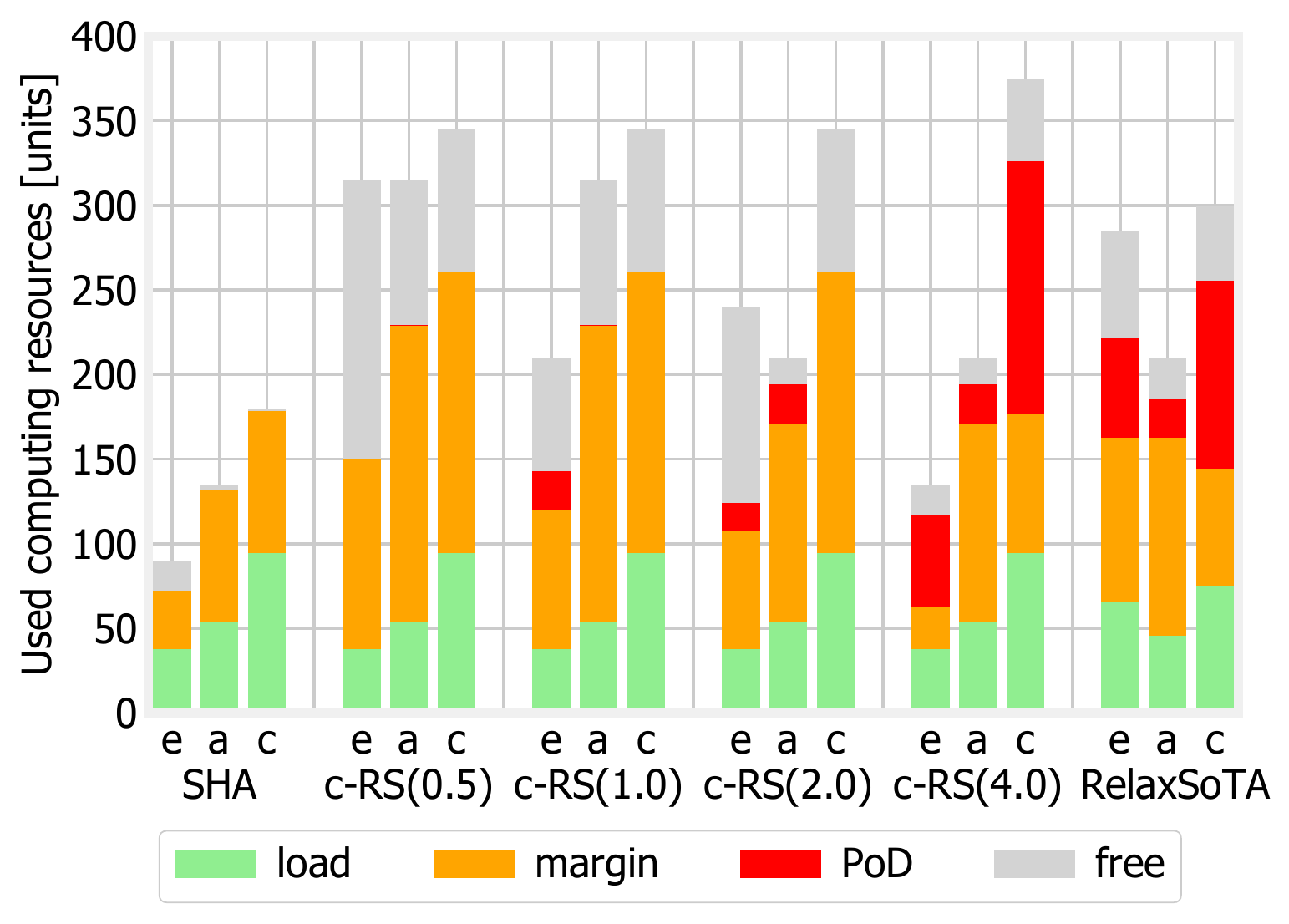}
\includegraphics[width=.32\textwidth]{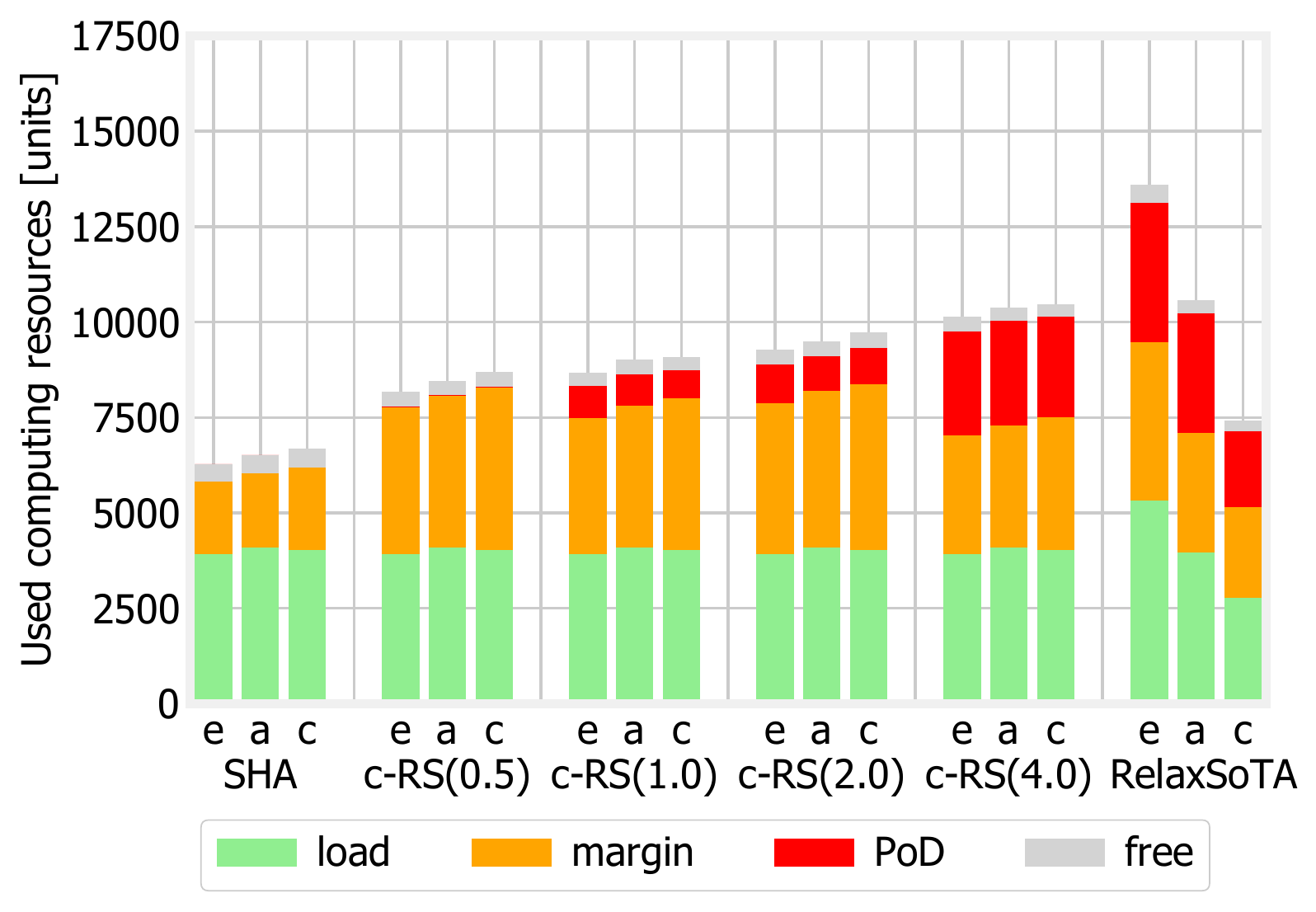}
\caption{
Three-layer scenario, smart-factory application. \highshare($\eps$), with different values of $\eps$ (labeled by c-RS($\eps$) for short): actual service
latency normalized to the target value (left);  usage of computing
resources for the first 15 requests (center) and all requests (right).
\label{fig:sf-details}
} 
\end{figure*}

\begin{figure*}
\centering
\includegraphics[width=.46\textwidth]{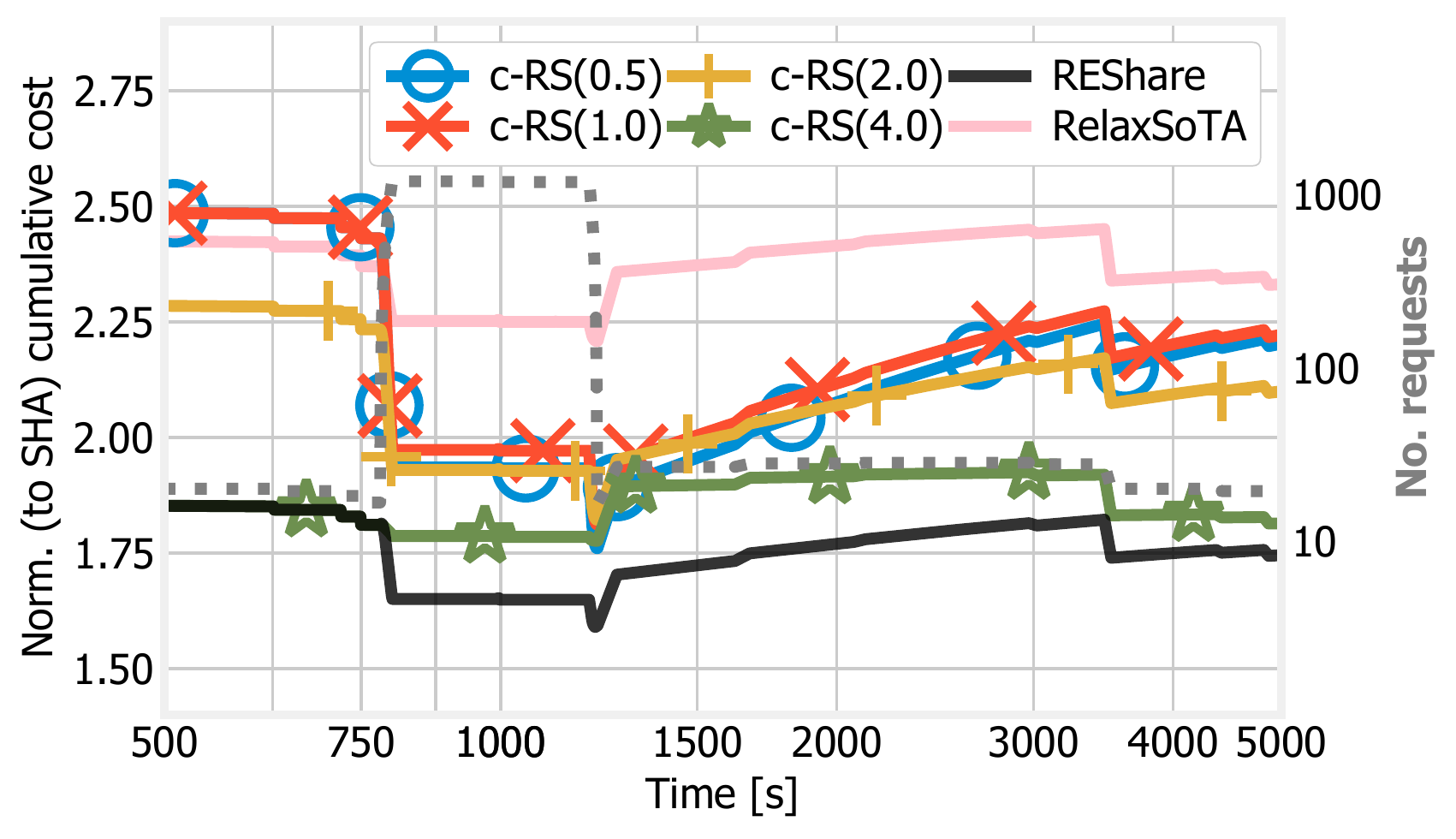}
\hspace{5mm}
\includegraphics[width=.46\textwidth]{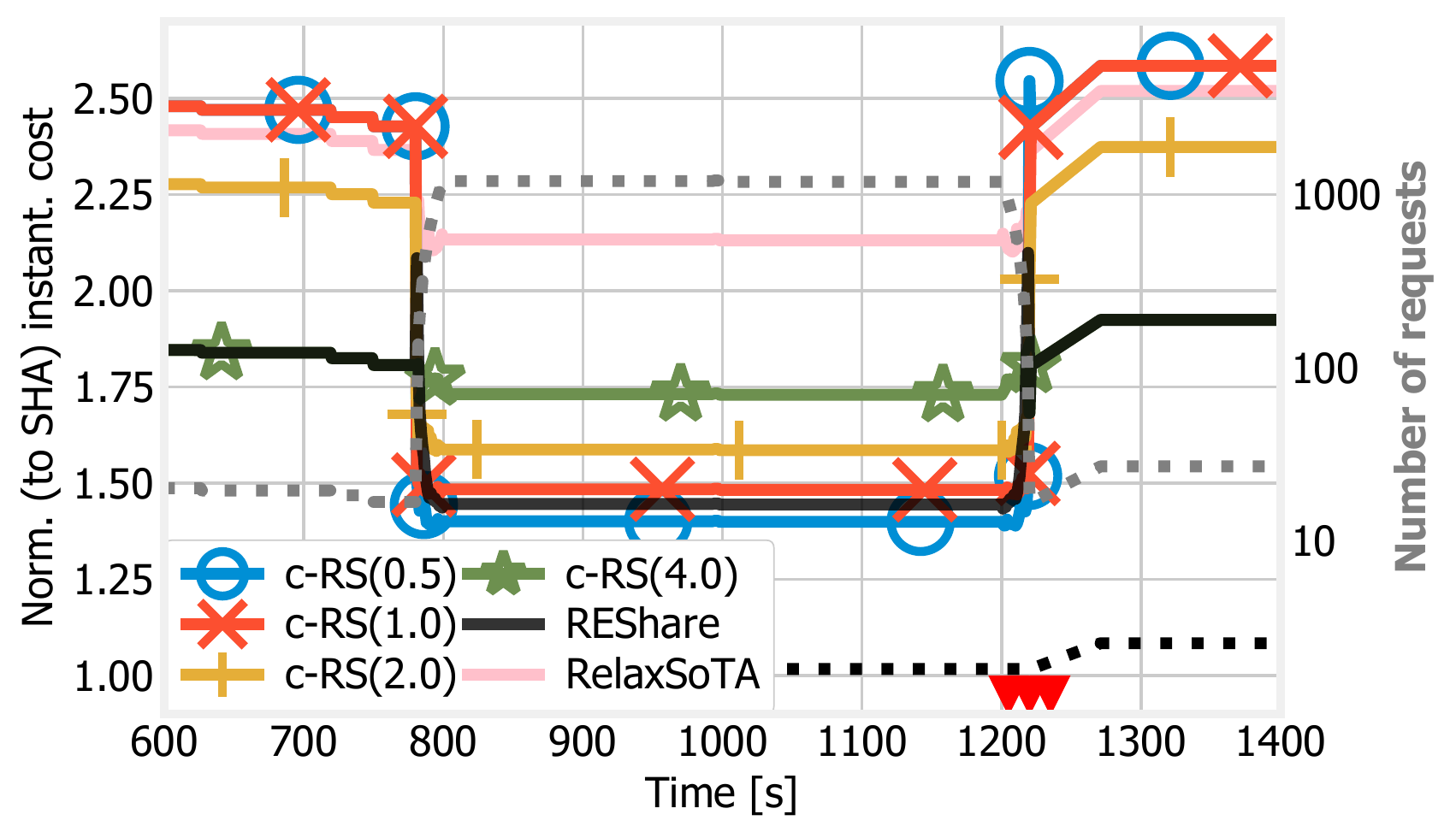}
\caption{
Four-layer scenario, real-world workload based upon the Materna trace~\cite{kohne2014federatedcloudsim,kohne2016evaluation}. \dhighshare\ and benchmark strategies: cumulative cost (left) and details of instantaneous cost during the load peak (right). In both plots, the dotted line corresponds to the load. In the right plot, upwards and downwards triangles at the bottom correspond to increasing and decreasing $\eps$ in \dhighshare.
\label{fig:real-cost}
} 
\centering
\includegraphics[width=.32\textwidth]{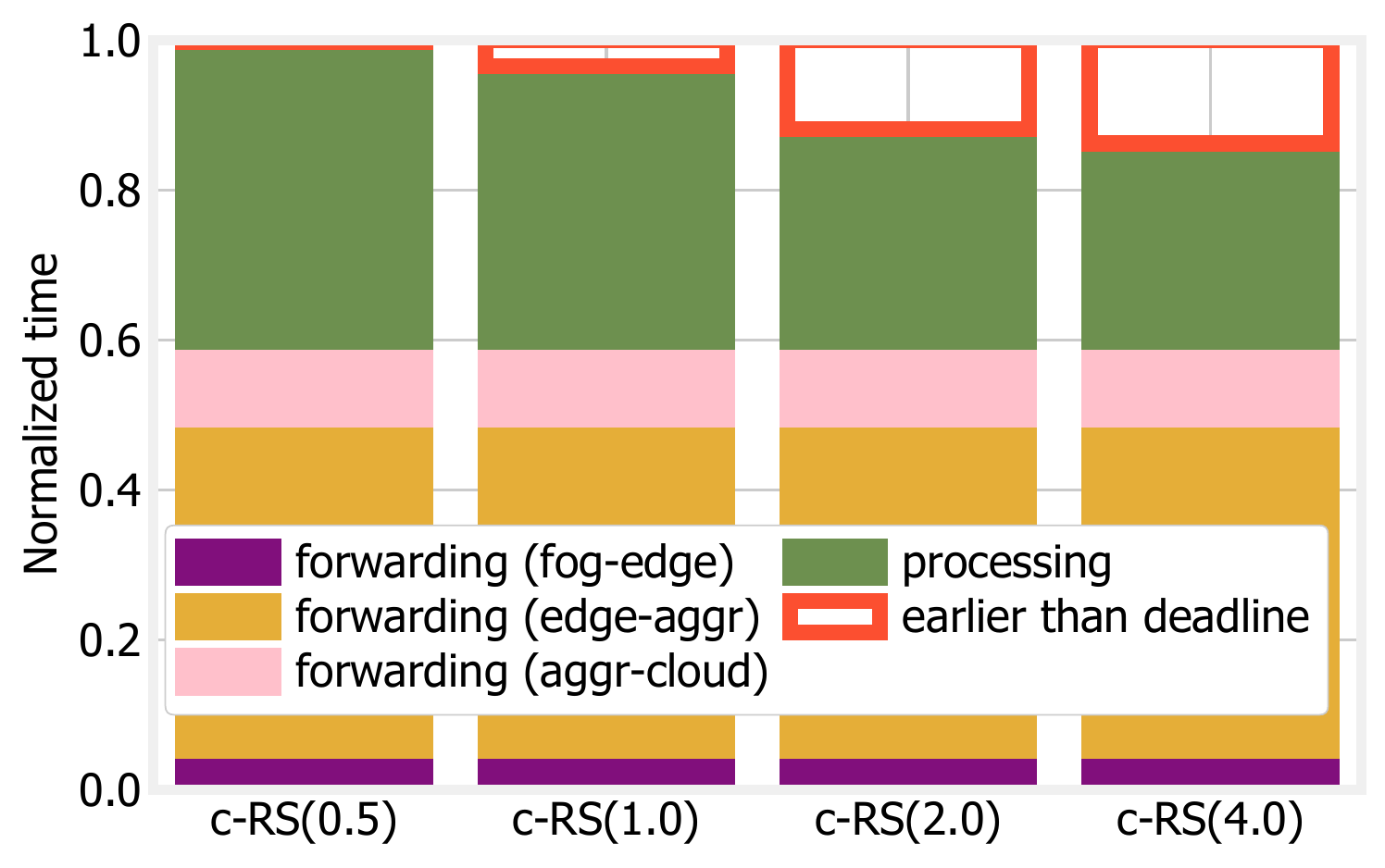}
\includegraphics[width=.32\textwidth]{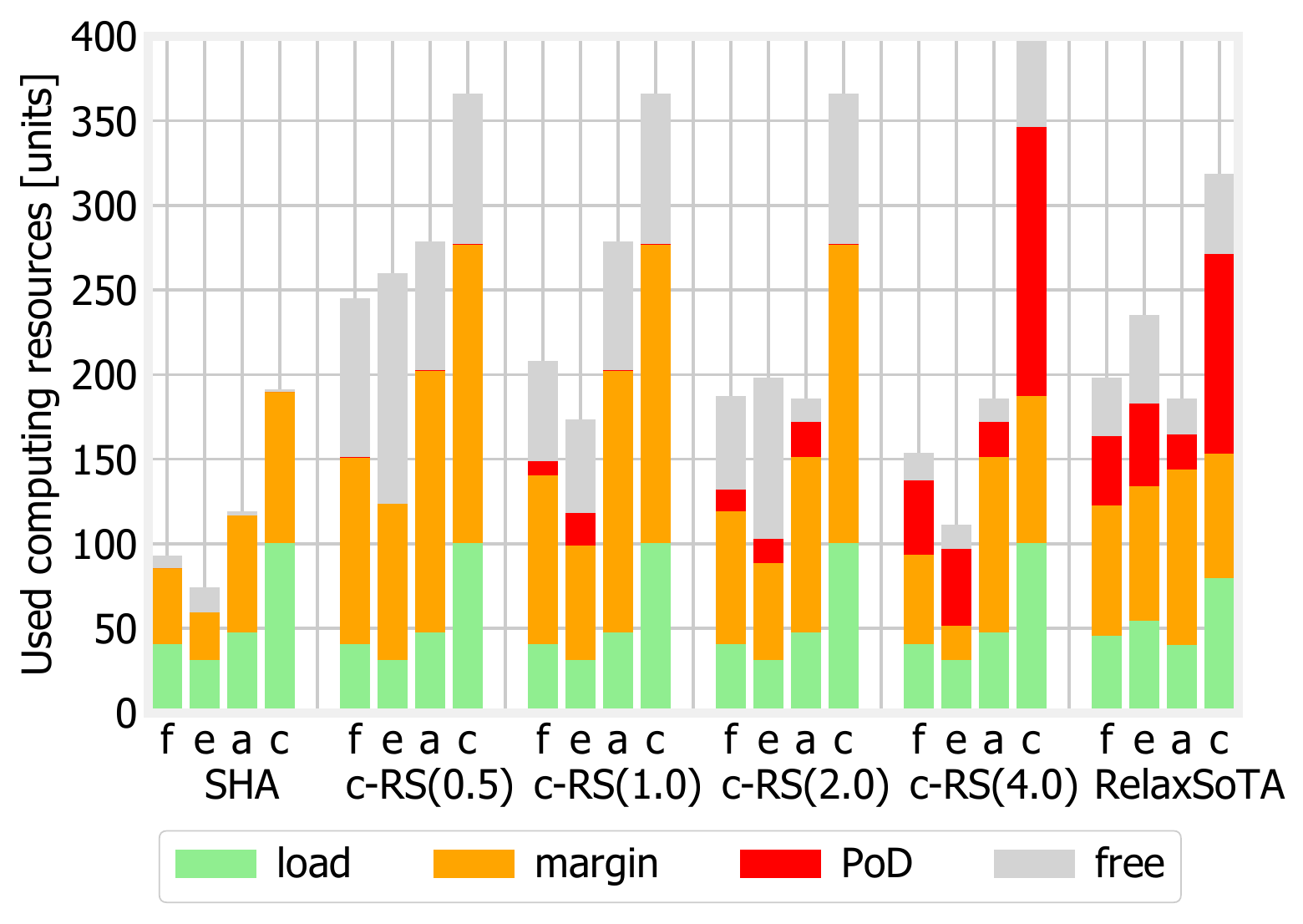}
\includegraphics[width=.32\textwidth]{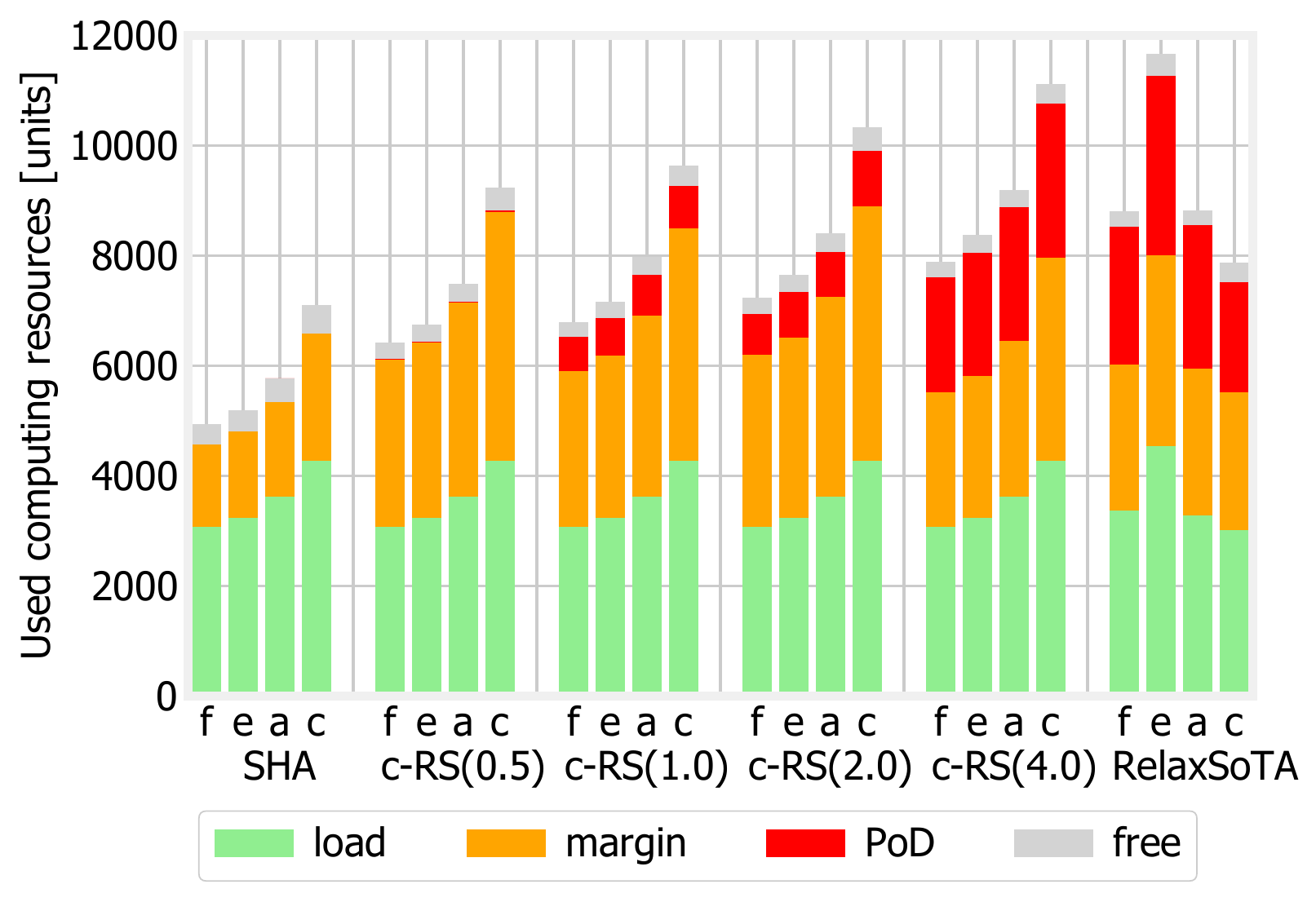}
\caption{
Four-layer scenario, real-world workload based on the Materna trace~\cite{kohne2014federatedcloudsim,kohne2016evaluation}. \highshare($\eps$) performance for different values of $\eps$ (labeled by c-RS($\eps$) for short): actual service
latency normalized to the target value (left);  usage of computing
resources for the first 15 requests (center), and all requests (right).
\label{fig:real-details}
} 
\end{figure*}

\subsection{Results}
\label{sec:sub-results}

\begin{table}[b]
\caption{
Cost savings brought by \dhighshare\ w.r.t. the state-of-the-art (RelaxSoTA)}\label{tab:savings}
\centering
\begin{tabularx}{0.8\columnwidth}{|X|r|} 
\hline
Scenario & Savings [\%]\\
\hline\hline
Vehicular/uniform & 15 \\
\hline
Smart factory & 24 \\
\hline
Materna (real-world) & 26 \\
\hline
\end{tabularx}
\end{table}

{\bf Vehicular domain.} 
A fundamental aspect to investigate is the cost of \dhighshare\ and its alternatives. In \Fig{cost}(left), we compare \dhighshare, \highshare($\eps$) with four different values of $\eps$ (identified by different markers), and the relaxation-based, state-of-the-art approach. For better readability, all values are normalized to the cost of SHA, 
which is why some cumulative costs in \Fig{cost}(right) appear to decrease. 
The gray, dotted line shows the number of requests in the system.

It can be immediately seen that \dhighshare\ is substantially cheaper
than the state-of-the-art benchmark we compare against and performs
close to \sha,
whose cost, we recall, is very close the optimum.
Furthermore, \dhighshare\ is 
cheaper than \highshare($ \eps $) for all $\eps$ values; this
confirms the effectiveness of the strategy implemented in
\Alg{dynamichighshare}, whereby the value of~$\eps$ is adjusted according to time-varying load conditions.

\Fig{cost}(right) presents the evolution of instantaneous (i.e., per
time unit) costs during and around the load peak. Note that low
values of $\eps$ are associated with a lower cost in high-traffic
periods, while larger $\eps$ values yield lower costs in
low-traffic periods, as highlighted in \Fig{PoD}. This observation is consistent with how $\eps$~values determine how much PoD we tolerate on each open VM. 
When only a few requests are present, it is cheaper to tolerate a high
PoD (large $\eps$) since, otherwise, we would open an excessive number
of (near empty) VMs. However, when the number of requests is high,
the larger $\eps$ implies that we utilize the open VMs inefficiently.  
As for \dhighshare, its cost is always close, albeit not equal, to the one of the {\em cheapest} instance of \highshare($\eps$) due to the switching behavior of \Alg{dynamichighshare}. Transitions between different  $\eps$ values are marked in \Fig{cost}(right) by upwards- and downwards-pointing triangles.

We now characterize how $\eps$ affects the system's performance and cost. \Fig{details}(left) shows how much of the services
delay budget ($D^s_r$)  is consumed by traffic forwarding
(yellow areas) and processing  (green areas). The traffic forwarding overheads are determined entirely by level choice to accommodate the job, which is the same for all $\eps$ values. However, larger values of $\eps$ result in shorter processing
times. Shorter processing times, i.e.,
providing services {\em earlier} than their constraint,
correspond to more computing resources unnecessarily provisioned and thus higher-than-needed costs.

\begin{table*}[]
\centering
\caption{Comparison of competitive approaches to VNF placement
\label{tab:prev_work_comparison}
} 
\begin{tabularx}{1\textwidth}{|X|l|l|l|X|l|} 
\hline
Description & Model/approach & Latency & Dynamic & Guarantees & Refs \\
\hline\hline
REShare & bin-packing, M/M/1 & yes & yes & constant asymptotic competitive ratio: 2 & our work \\ 
\hline
Capacitated NFV Location Algorithm & Generalized Assignment Problem (GAP) & no & no & bi-criteria: cost is at most 8 times the optimum, and constraints are violated by a factor of at most 8 &\cite{RCohen15}\\
\hline
SPR$^3$ & multi-dimensional optimization; randomized approach &  yes &  no &  competitive ratio (with high probability) of $4+\frac{3\log S}{R_n}$, where $S$ and $R_n$ are instance-dependent factors; constraints are satisfied in expectation & \cite{tulino2}\\
\hline
JASPER & multi-dimensional optimization; randomized approach & yes & no &  competitive ratio (with high probability) of $3+\frac{2\log S}{\xi^\dagger}$, where $S$ and $\xi^\dagger$ are instance-dependent factors; constraints are violated, with high probability, by at most a factor $4+\delta$, with $\delta \geq 0$ being a scenario-dependent quantity & \cite{draxler2018jasper}\\
\hline
GFT & MILP optimization & yes & no & asymptotic competitive ratio: $2+(1-o(1))\log m$, where $m$ is instance-dependent & \cite{sang2017provably}\\
\hline
QNSD & multi-commodity-chain flow (MCCF) & yes & no & none for the full (integer) problem; $O(\eps)$ competitive ratio for the fractional (relaxed) one & \cite{feng2017approximation}\\
\hline
GSP-GRS & MILP optimization & yes & no & $2$ in special cases, none in general scenarios & \cite{he2018s}\\
\hline
GSP-SS & multi-scale scheduling &  yes & yes & none & \cite{farhadi2021service}\\
\hline
Online Throughput Maximization Algorithm & MILP optimization & yes & yes & $O(\log n)$, where $n$ is instance-dependent & \cite{ma2019throughput}\\
\hline
\end{tabularx}
\end{table*}

This is confirmed by \Fig{details}(center) and \Fig{details}(right),
highlighting how the VMs capacity is used. Green areas therein
correspond to the load VMs have to serve,
which
cannot be
reduced. The sum of orange and red areas correspond to the
margin~$\mu_b-\theta_v\Lambda(b)$ of the VMs (see \Sec{algo}); such a
quantity must be larger than~$\frac{1}{D_r^v}$ for all requests served
by each VM. In particular, the orange areas correspond to the margin
that VMs would have {\em if all requests they serve had the same
  latency constraints}, while the red ones correspond to the
PoD. Finally, gray areas correspond to the
difference~$\bar{\mu}-\mu_b$ between the maximum and actually allocated VM capacities.

We can see that larger $\eps$ values are always associated with a higher PoD. 
If the load is low (as in \Fig{details}(center)), larger $\eps$ values, implying more VM sharing and a higher PoD, may be an acceptable alternative to provisioning more VMs. This observation explains the behavior we observed in \Fig{cost},
where larger values of $\eps$ result in lower cost {\em in spite} of a higher PoD.
For high load (\Fig{details}(right)),
limiting the PoD is instrumental in reducing the quantity of consumed
resources. 
Specifically, from \Fig{details}(right), we can see that the PoD is over 5\% for RelaxSoTA, while it drops below 1\% for \dhighshare.

{\bf Smart-factory domain.}
\Fig{sf-cost} and \Fig{sf-details} present the performance of \dhighshare\ and its alternatives for the smart-factory application.
\Fig{sf-cost} confirms that \dhighshare\ yields the lowest {\em cumulative} cost (left plot), despite not necessarily being the cheapest solution at every point in time (right plot).
It is also interesting to notice how the faster pace at which the load evolves also implies that \highshare\ with low values of~$\eps$ cannot catch up with \dhighshare, and yield a substantially higher cumulative cost (first plot). \dhighshare, on the other hand, can quickly go through all $\eps$~values and reach the optimal one, as shown by the green and red triangles at the bottom of the second plot of \Fig{sf-cost}.

\Fig{sf-details} shows how, despite the different services, smaller values of~$\eps$ are consistently associated with a smaller PoD, though not necessarily with the lowest cost.
By quickly reaching the right value of~$\eps$, \dhighshare\ keeps the PoD below 1\%, compared to 13\% of RelaxSoTA.

{\bf Materna workload.}
The results for the real-world scenario based on the Materna trace are summarized in  \Fig{real-cost} and \Fig{real-details}. We can observe a behavior that is effectively equivalent to \Fig{sf-cost} and \Fig{sf-details}, which further confirms how \dhighshare\ works well with different loads and network topologies. As we can see from \Fig{real-details}(right), the PoD for  \dhighshare\ is below 1\%, compared to 21\% of the state-of-the-art solutions. 
Throughout all scenarios, using \dhighshare\ {\em in lieu} of state-of-the-art approaches consistently yields very significant cost savings, as summarized in \Tab{savings}. 
Interestingly, savings are higher in more complex scenarios, e.g., the real-world one.

\section{Related Work\label{sec:rel-work}}

The pioneering  work on VNF placement~\cite{RCohen15}
casts the problem into a generalized assignment problem (GAP), and
proposes a bi-criteria approximation thereto.  Recent
works~\cite{feng2017approximation,noi-satyam,poularakis2019joint,noi-jorge,flexshare,tulino}
widen the focus of the orchestration problem to include traffic
routing as well as VNF placement. These studies present non-linear (and
non-convex) problem formulations and, thus, resort to heuristics to
 solve the resulting problem. Other popular methodologies
include graph theory~\cite{ma2017traffic,draxler2018jasper} and
set-covering~\cite{sang2017provably}. Several works also account for 
 VNFs performing, 
the fact that VNFs can perform multiple tasks,
 e.g., caching \cite{xu2018joint,chen2018edge}. 
\cite{he2018s}~follows a similar approach and jointly solves the problems of VNF placement and scheduling, i.e., which physical resources to use and when. In the same setting, \cite{farhadi2021service}~makes placement and scheduling decisions accounting for multiple resources, e.g., memory and storage, so as to reflect the requirements of the existing VNFs. Other works focus on specific services, e.g., \cite{ma2019throughput}~considers multicast streaming in MEC scenarios, and its peculiar requirements in terms of network latency and VNF capacity.

Most schemes work offline, i.e., all the service instances are known in advance. Among the few online approaches that
deal with requests arriving at different times,
\cite{lukovszki2016s,lukovszki2018approximate}   
incrementally
update the current configuration, minimizing the changes to
accommodate the new requests.
In a similar setting,
\cite{tulino,tulino2}~process service requests via a {\em randomized} approach.
More recently,
\cite{blocher2020letting}~performed placement  offline and routing  online.

While many works account for the fact that individual hosts (e.g.,
VMs) may have different capabilities and features, few consider
layered topologies. Among those,
\cite{guo2015shadow,poularakis2019joint}~focus on the choice between
edge and cloud resources, and
\cite{cohen2019access}~studies the same problem with reference to
caching, while \cite{kamran2019deco}~aims at  jointly
placing the VNFs and the data they need. 
Finally, several works  characterize or predict service requests' arrival, thereby  simplifying network management. Approaches include exploiting the traffic variability to reduce the amount of needed resources~\cite{bouet2018mobile}, using reinforcement learning to predict traffic~\cite{sciancalepore2018z}, and estimating the resources needed by each request before admitting it~\cite{han2019utility}.
Although \dhighshare\ does not {\em require} any knowledge about the future evolution of the time demand, such information can be exploited, when available, to further improve its performance.

It is important to stress that, unlike \dhighshare\ and \highshare, existing works~\cite{tulino,tulino2,moualla2019online,liu2017dynamic} assume that VNF requirements are constant over time, and either are or are not satisfied by VMs.
In \Tab{prev_work_comparison}, we provide a summary of the comparison between previous work studying competitive approaches to VM placement, and our proposed solution. Importantly, ours is the {\em only} work featuring both a constant competitive ratio and the ability to deal with dynamic scenario, i.e., time-varying service demand.

A research problem closely related, albeit orthogonal, to \dhighshare\ is represented by {\em predicting} the future demand for content and services. Examples include~\cite{bega2019deepcog}, where the authors envision using a deep neural network (DNN) to forecast the future demand, and make network orchestration decisions based upon such a forecast. In a similar spirit, \cite{xu2022esdnn}~explores novel DNN architectures to better predict the load of cloud applications. It is important to stress that, whenever available, such predictions can seamlessly be integrated within \dhighshare\ and further boost its performance.

\section{Conclusions and future work}\label{sec:conclusions}

We addressed the problem of creating and scaling network slices while trying to reuse existing (sub-)slices across different services. 
We considered the availability of resources at different locations, including edge, aggregation,
and cloud, and a time-varying system workload
with service requests, arrivals, and departures.
To effectively create and scale sub-slices,
we proposed a low-complexity algorithm, which we proved to be asymptotic 2-competitive in the case of a non-decreasing load. Furthermore, numerical results obtained considering real-world services showed that our solution outperforms alternative approaches, 
for {\em time-varying workloads},
reducing the service cost by over 25\%.

Future work will focus on extending our system model, most notably, by considering: (i) workload prediction as an approach to resource allocation for time-varying workloads, and (ii) VM migration, i.e., the possibility that VNFs move across different nodes during the service lifetime, and the consequent need  to balance the activation of fewer VMs against  migration cost. 
Furthermore, we will seek to implement \dhighshare\ and assess its performance in real-world scenarios, first in small-scale testbeds and then through larger, cloud-based experiments. By so doing, we will be able to better prove the practical effectiveness of our solution.

\bibliographystyle{IEEEtran}
\bibliography{refs}

\end{document}